\documentclass[letterpaper, 10 pt, conference]{ieeeconf}  

\IEEEoverridecommandlockouts                              
\usepackage{amsmath,amsfonts,amssymb,url,listings,mathrsfs,wasysym}
\usepackage{enumerate}
\usepackage{graphicx}
\usepackage{subfigure}
\usepackage{paralist}
\usepackage[version=0.96]{pgf}
\usepackage{tikz}
\usetikzlibrary{shapes,arrows,automata}

\newcommand{\num}[1]{\relax\ifmmode \mathbb #1\else $\mathbb #1$\fi}

\newcommand{\nnreals}{{\num R}_{\geq 0}}

\newcommand{\oor}{\vee}
\newcommand{\aand}{\wedge}
\newcommand{\until}{\hspace{1mm}\mathcal{U}\hspace{1mm}}
\newcommand{\always}{\square}
\newcommand{\eventually}{\Diamond}
\renewcommand{\next}{\ocircle}
\newcommand{\true}{\mathit{True}}

\newcommand{\A}{{\mathcal{A}}}
\newcommand{\C}{{\mathcal{C}}}
\newcommand{\D}{{\mathcal{D}}}

\newcommand{\lang}{\mathcal{L}_\omega}
\newcommand{\M}{{\mathcal{M}}}

\newcommand{\T}{{\mathcal{T}}}
\newcommand{\U}{{\mathcal{U}}}
\newcommand{\V}{{\mathcal{V}}}

\newcommand{\prob}{\mathrm{Pr}}

\newtheorem{lemma}{Lemma}
\newtheorem{proposition}{Proposition}
\newtheorem{corollary}{Corollary}
\newtheorem{definition}{Definition}
\newtheorem{remark}{Remark}
\newtheorem{example}{Example}
\newtheorem{problem}{Problem}

\title{\LARGE \bf
Control of Probabilistic Systems under Dynamic, Partially Known Environments
with Temporal Logic Specifications
}
\author{Tichakorn Wongpiromsarn and Emilio Frazzoli\vspace{-5mm}
\thanks{Tichakorn Wongpiromsarn is with the Singapore-MIT Alliance for Research and Technology,
Singapore 117543, Singapore. {\tt\footnotesize nok@smart.mit.edu}}
\thanks{Emilio  Frazzoli is with the Massachusetts Institute of Technology, Cambridge, MA 02139, USA.
  {\tt\footnotesize frazzoli@mit.edu}}%
}

\begin{document}
\maketitle
\thispagestyle{empty}
\pagestyle{empty}

\begin{abstract}
We consider the synthesis of control policies for probabilistic systems, modeled by Markov
decision processes, operating in partially known environments with temporal logic specifications.
The environment is modeled by a set of Markov chains.
Each Markov chain describes the behavior of the environment in each mode.
The mode of the environment, however, is not known to the system.
Two control objectives are considered:
maximizing the expected probability and maximizing the worst-case probability 
that the system satisfies a given specification.
\end{abstract}

\section{Introduction}
In many applications, control systems need to perform complex tasks and interact with their (potentially adversarial) environments.
The correctness of these systems typically depends on the behaviors of the environments. 
For example, whether an autonomous vehicle exhibits a correct behavior at a pedestrian crossing depends on
the behavior of the pedestrians, e.g., whether they actually cross the road, remain on the same side of the road
or step in front of the vehicle while it is moving.

Temporal logics, which were primarily developed by the formal methods community 
for specifying and verifying correctness of software and hardware systems, 
have been recently employed to express complex behaviors of control systems.
Its expressive power offers extensions to properties that can be expressed
than safety and stability, typically studied in the control and hybrid systems domains.
In particular, \cite{WKF11-ITSC} shows that the traffic rule enforced in the 2007 DARPA Urban Challenge can be precisely described using these logics.
Furthermore, the recent development of language equivalence and simulation notions allows abstraction of continuous systems to a purely discrete model \cite{alur00,tanner02,Girard09Hierarchical}.
This subsequently provides a framework for integrating methodologies from the formal methods and the control-theoretic communities and enables formal specification, design and verification of control systems with complex behaviors. 

Controller synthesis from temporal logic specifications has been considered in
\cite{Tabuada04lineartime,Kloetzer08,karaman09cdc,Bhatia:SMP2010}, assuming static environments.
Synthesis of reactive controllers that takes into account all the possible behaviors of dynamic environments 
can be found in \cite{KGFP_TRO09,wongpiromsarn10hscc}.
In this case, the environment is treated as an adversary and the synthesis problem can be viewed as a two-player game 
between the system and the environment: the environment attempts to falsify the specification while the system attempts to satisfy it
\cite{piterman06}.
In these works, the system is assumed to be deterministic, i.e., an available control action in each state enables exactly one transition.
Controller synthesis for probabilistic systems such as Markov decision processes (MDP) has been considered in \cite{Ding:LTL2011,Ding:MDP2011}.
However, these works assume that at any time instance, the state of the system, including the environment, as well as
their models are fully known.
This may not be a valid assumption in many applications.
For example, in the pedestrian crossing problem previously described, the behavior of the pedestrians depend on their destination,
which is typically not known to the system.
Having to account for all the possible behaviors of the pedestrians with respect to all their possible destinations 
may lead to conservative results and in many cases, unrealizable specifications.

Partially observable Markov decision process (POMDP) provides a principled mathematical framework to cope with 
partial observability in stochastic domains \cite{Kaelbling:POMDP1998}.
Roughly, the main idea is to maintain a belief, which is defined as a probability distribution over all the possible states.
POMDP algorithms then operate in the belief space.
Unfortunately, it has been shown that solving POMDPs exactly is computationally intractable \cite{Papadimitriou:CMD1987}.
Hence, point-based algorithms have been developed to compute an approximate solution based on
the computation over a representative set of points from the belief space rather than the entire belief space 
\cite{Hauskrecht:VAP2000}.

In this paper, we take an initial step towards solving POMDPs that are subject to temporal logic specifications.
In particular, we consider the problem where a collection of possible environment models is available to the system.
Different models correspond to different modes of the environment.
However, the system does not know in which mode the environment is.
In addition, the environment may change its mode during an execution subject to certain constraints.
We consider two control objectives: maximizing the expected probability and maximizing the worst-case probability 
that the system satisfies a given temporal logic specification.
The first objective is closely related to solving POMDPs as previously described whereas
the second objective is closely related to solving uncertain MDPs \cite{Nilim:RCM2005}.
However, for both problems, we aim at maximizing the expected or worst-case probability of satisfying
a temporal logic specification, instead of maximizing the expected or worst-case reward as considered in the POMDP and MDP literature.

The main contribution of this paper is twofold.
First, we show that the expectation-based synthesis problem can be formulated 
as a control policy synthesis problem for MDPs under temporal logic specifications
and provide a complete solution to the problem.
Second, we define a mathematical object called \emph{adversarial Markov decision process} (AMDP) and show
that the worst-case-based synthesis problem can be formulated as a control policy synthesis problem for AMDP.
A complete solution to the control policy synthesis for AMDP is then provided.
Finally, we show that the maximum worst-case probability that a given specification is satisfied
does not depend on whether the controller and the adversary play alternatively or 
both the control and adversarial policies are computed at the beginning of an execution.

The rest of the paper is organized as follows: 
We provide useful definitions and descriptions of the formalisms in the following section. 
Section \ref{sec:prob} is dedicated to the problem formulation.
The expectation-based and the worst-case-based control policy synthesis are considered in
Section \ref{sec:syn-MDP} and Section \ref{sec:syn-AMDP}, respectively.
Section \ref{sec:ex} presents an example.
Finally, Section \ref{sec:conclusions} concludes the paper and discusses future work.

\section{Preliminaries}
We consider systems that comprise stochastic components.
In this section, we define the formalisms used in this paper to describe such systems and their desired properties.
Throughout the paper, we let $X^*$, $X^\omega$ and $X^+$ denote the set of 
finite, infinite and nonempty finite strings, respectively, of a set $X$.

\subsection{Automata}
\begin{definition} 
A \emph{deterministic Rabin automaton}  (DRA) is a tuple
$\A = (Q, \Sigma, \delta, q_{init}, Acc)$ where
\begin{itemize}
\item $Q$ is a finite set of states,
\item $\Sigma$ is a finite set called alphabet,
\item $\delta: Q \times \Sigma \to Q$ is a transition function,
\item $q_{init} \in Q$ is the initial state, and
\item $Acc \subseteq 2^Q \times 2^Q$ is the acceptance condition.
\end{itemize}
We use the relation notation, $q \stackrel{w}{\longrightarrow} q'$
to denote $\delta(q,w) = q'$.
\end{definition}

Consider an infinite string $\sigma = \sigma_0\sigma_1 \ldots \in \Sigma^\omega$.
A \emph{run} for $\sigma$
in a DRA $\A = (Q, \Sigma, \delta, q_{init}, Acc)$ is an infinite sequence of states
$q_0q_1\ldots q_n$ such that $q_0 = q_{init}$ and 
$q_i \stackrel{\sigma_i}{\longrightarrow} q_{i+1}$ for all $i \geq 0$.
A run is \emph{accepting} if there exists a pair $(H,K) \in Acc$ such that
(1) there exists $n \geq 0$ such that for all $m \geq n$, $q_m \not\in H$, and
(2) there exist infinitely many $n \geq 0$ such that $q_n \in K$.

A string $\sigma \in \Sigma^*$ is \emph{accepted} by $\A$ if there is an accepting run of $\sigma$ in $\A$.
The language \emph{accepted} by $\A$, denoted by $\lang(\A)$, is the set of all accepted strings of $\A$.

\subsection{Linear Temporal Logic}
Linear temporal logic (LTL) is a branch of logic that can be used to reason about a time line.
An LTL formula is built up from a set $\Pi$ of atomic propositions,
the logic connectives $\neg$, $\oor$, $\aand$ and $\Longrightarrow$
and the temporal modal operators 
$\next$ (``next''), $\always$ (``always''), $\eventually$ (``eventually'') and $\until$ (``until'').
An LTL formula over a set $\Pi$ of atomic propositions is inductively defined as
\[
  \varphi := \true \hspace{1mm}|\hspace{1mm} p \hspace{1mm}|\hspace{1mm} \neg \varphi \hspace{1mm}|\hspace{1mm}
   \varphi \aand \varphi \hspace{1mm}|\hspace{1mm} \next \varphi \hspace{1mm}|\hspace{1mm} \varphi \until \varphi 
\]
where $p \in \Pi$.
Other operators can be defined as follows: 
$\varphi \aand \psi = \neg(\neg\varphi \oor \neg\psi)$, 
$\varphi \Longrightarrow \psi = \neg \varphi \oor \psi$,
$\eventually\varphi = \true \until \varphi$, and
$\always\varphi = \neg\eventually\neg\varphi$.

\textit{\textbf{Semantics of LTL:}}
LTL formulas are interpreted on infinite strings over $2^{\Pi}$.
Let $\sigma = \sigma_0 \sigma_1 \sigma_2 \ldots$ where $\sigma_i \in 2^{\Pi}$ for all $i \geq 0$.
The	satisfaction relation $\models$	is defined inductively on LTL formulas as follows:
\begin{itemize}
\item $\sigma \models \true$,
\item for an atomic proposition $p \in \Pi$, $\sigma \models p$ if and only if $p \in \sigma_0$,
\item $\sigma \models \neg \varphi$ if and only if $\sigma \not\models \varphi$,
\item $\sigma \models \varphi_1 \aand \varphi_2$ if and only if $\sigma \models \varphi_1$ and $\sigma \models \varphi_2$,
\item $\sigma \models \next\varphi$ if and only if $\sigma_1 \sigma_2\ldots \models \varphi$, and
\item $\sigma \models \varphi_1 \until \varphi_2$ if and only if there exists $j \geq 0$ such that 
$\sigma_j \sigma_{j+1} \ldots \models \varphi_2$ and for all $i$ such all $0 \leq i < j$, $\sigma_i \sigma_{i+1}\ldots \models \varphi_1$.
\end{itemize}

Given propositional formulas $p_1$ and $p_2$,
examples of widely used LTL formulas include 
a safety formula $\always p_1$ (read as ``always $p_1$''), which simply asserts that property $p_1$ remains invariantly true throughout an execution
and 
a reachability formula $\eventually p_1$ (read as ``eventually $p_1$''), which states that property $p_1$ becomes true at least once in an execution
(i.e., there exists a reachable state that satisfies $p_1$).
In the example presented later in this paper, we use a formula $p_1 \until p_2$ (read as ``$p_1$ until $p_2$''), which
asserts that $p_1$ has to remain true until $p_2$ becomes true and there is some point in an execution where $p_2$ becomes true.

It can be shown that for any LTL formula $\varphi$ over $\Pi$, 
there exists a DRA $\A$ with alphabet $\Sigma = 2^\Pi$ that accepts all and only words over $\Pi$ that satisfy $\varphi$,
i.e., $\lang(\A) = \{\sigma \in (2^\Pi)^\omega \hspace{1mm}|\hspace{1mm} \sigma \models \varphi\}$.
Such $\A$ can be automatically constructed using existing tools \cite{Klein:EDF2006}.
We refer the reader to \cite{Baier:PMC2008,Allen:TML90,zohar92} for more details on LTL.

\subsection{Systems and Control Policies}
\label{ssec:policies}

\begin{definition}
	A \emph{(discrete-time) Markov chain} (MC) is a tuple $\M = (S, \mathbf{P}, s_{init}, \Pi, L)$ where
	\begin{itemize}
	\item $S$ is a countable set of states,
	\item $\mathbf{P} : S \times S \to [0,1]$ is the transition probability function such that for any state $s \in S$,
	$\sum_{s' \in S} \mathbf{P}(s, s') = 1$,
	\item $s_{init} \in S$ is the initial state,
	\item $\Pi$ is a set of atomic propositions, and
	\item $L : S \to 2^\Pi$ is a labeling function.
	\end{itemize}
\end{definition}

\begin{definition}
	A \emph{Markov decision process} (MDP) is a tuple
	$\M = (S, Act, \mathbf{P}, s_{init}, \Pi, L)$ where
	$S$, $s_{init}$, $\Pi$ and $L$ are defined as in MC and
	\begin{itemize}
	\item $Act$ is a finite set of actions, and
	\item $\mathbf{P} : S \times Act \times S \to [0,1]$ is the transition probability function such that for any state $s \in S$ and action $\alpha \in Act$,
	$\sum_{s' \in S} \mathbf{P}(s, \alpha, s') \in \{0, 1\}$ 
	\end{itemize}

	An action $\alpha$ is \emph{enabled} in state $s$ if and only if $\sum_{s' \in S} \mathbf{P}(s, \alpha, s') = 1$.
	Let $Act(s)$ denote the set of enabled actions in $s$.
\end{definition}

Given a complete system as the composition of all its components, we are interested in computing
a control policy for the system that optimizes certain objectives.
We define a control policy for a system modeled by an MDP as follows.

\begin{definition}
	Let $\M = (S, Act, \mathbf{P}, s_{init}, \Pi, L)$ be a Markov decision process.
	A \emph{control policy} for $\M$ is a function $\C : S^+ \to Act$
	such that $\C(s_0s_1\ldots s_n) \in Act(s_n)$ for all $s_0s_1\ldots s_n \in S^+$.
\end{definition}

Let $\M = (S, Act, \mathbf{P}, s_{init}, \Pi, L)$ be an MDP
and $\C : S^+ \to Act$ be a control policy for $\M$.
An infinite sequence $r_\M^\C = s_0 s_1 \ldots$ on $\M$ generated under policy $\C$
is called a \emph{path} on $\M$ if $s_0 = s_{init}$ and
$\mathbf{P}(s_i, \C(s_0s_1\ldots s_i), s_{i+1}) > 0$ for all $i$.
The subsequence $s_0s_1 \ldots s_n$ where $n \geq 0$ is
the \emph{prefix} of length $n$ of $r_\M^\C$.
We define $Paths_\M^\C$ and $FPaths_\M^\C$ as the set of all infinite paths of $\M$
under policy $\C$ and their finite prefixes, respectively.
For $s_0s_1\ldots s_n \in FPaths_\M^\C$, we let $Paths_\M^\C(s_0s_1\ldots s_n)$
denote the set of all paths in $Paths_\M^\C$ with prefix $s_0s_1 \ldots s_n$.

The $\sigma$-algebra associated with $\M$ under policy $\C$
is defined as the smallest $\sigma$-algebra that contains $Paths_\M^\C(\hat{r}_\M^\C)$
where $\hat{r}_\M^\C$ ranges over all finite paths in $FPaths_\M^\C$.
It follows that there exists a unique probability measure $Pr_\M^\C$ on the $\sigma-$algebra 
associated with $\M$ under policy $\C$ where for any $s_0s_1\ldots s_n \in FPaths_\M^\C$,
\begin{equation}
\label{eq:prob-measure}
  \begin{array}{l}
  \prob_\M^\C\{Paths_\M^\C(s_0s_1\ldots s_n)\} = \\
  \hspace{10mm}\prod_{0 \leq i < n} \mathbf{P}(s_i, \C(s_0s_1\ldots s_i), s_{i+1}).
  \end{array}
\end{equation}

Given an LTL formula $\varphi$, one can show that the set 
$\{s_0 s_1 \ldots \in Paths_\M^\C \hspace{1mm}|\hspace{1mm} L(s_0) L(s_1) \ldots \models \varphi\}$
is measurable \cite{Baier:PMC2008}.
The probability for $\M$ to satisfy $\varphi$ under policy $\C$ is then defined as
\begin{equation*}
  \prob_\M^\C(\varphi) = \prob_\M^\C\{s_0 s_1 \ldots \in Paths_\M^\C \hspace{1mm}|\hspace{1mm}
  L(s_0) L(s_1) \ldots \models \varphi\} .
\end{equation*}

For a given (possibly noninitial) state $s \in S$, we let $\M^s = (S, Act, \mathbf{P}, s, \Pi, L)$,
i.e., $\M^s$ is the same as $\M$ except that its initial state is $s$.
We define $\prob_\M^\C(s \models \varphi) = \prob_{\M^s}^\C(\varphi)$ 
as the probability for $\M$ to satisfy $\varphi$ under policy $\C$, starting from $s$.

A control policy essentially resolves all the nondeterministic choices in an MDP and induces a Markov chain $\M_\C$
that formalizes the behavior of $\M$ under control policy $\C$ \cite{Baier:PMC2008}.
In general, $\M_\C$ contains all the states in $S^+$ and hence may not be finite even though $\M$ is finite.
However, for a special case where $\C$ is a memoryless or a finite memory control policy, 
it can be shown that $\M_\C$ can be identified with a finite MC.
Roughly, a memoryless control policy always picks the action based only on the current state of $\M$,
regardless of the path that led to that state.
In contrast, a finite memory control policy also maintains its ``mode'' and picks the action based on its current mode 
and the current state of $\M$.

\section{Problem Formulation}
\label{sec:prob}
Consider a system that comprises 2 components: the plant and the environment.
The system can regulate the state of the plant but
has no control over the state of the environment.
We assume that at any time instance, the state of the plant and the environment can be
precisely observed.

The plant is modeled by a finite MDP $\M^{pl} = (S^{pl}, Act, \mathbf{P}^{pl}, s_{init}^{pl}, \Pi^{pl}, L^{pl})$.
We assume that for each $s \in S^{pl}$, there is an action $\alpha \in Act$ that is enabled in state $s$.
In addition, we assume that the environment can be modeled by some MC in
$\mathbf{M}^{env} = \{\M^{env}_1, \M^{env}_2, \ldots, \M^{env}_N\}$
where for each $i \in \{1, \ldots, N\}$, 
$\M^{env}_i =  (S^{env}_i, \mathbf{P}^{env}_i, s_{init, i}^{env}, \Pi^{env}_i, L^{env}_i)$
is a finite MC that represents a possible model of the environment.
For the simplicity of the presentation, we assume that for all $i \in \{1, \ldots, N\}$,
$S^{env}_i  = S^{env}$, $\Pi^{env}_i = \Pi^{env}$, $s_{init,i}^{env} = s_{init}^{env}$ and $L^{env}_i = L^{env}$;
hence, $\M^{env}_1, \M^{env}_2, \ldots, \M^{env}_N$ differ only in the transition probability function.
These different environment models can be considered as different \emph{modes} of the environment.
For the rest of the paper, we use ``environment model'' and ``environment mode'' interchangeably
to refer to some $\M^{env}_i \in \mathbf{M}^{env}$.

We further assume that the plant and the environment make a transition simultaneously, i.e., 
both of them makes a transition at every time step.
All $\M^{env}_i \in \mathbf{M}^{env}$ are available to the system.
However, the system does not know exactly which $\M^{env}_i \in \mathbf{M}^{env}$ is the actual model of the environment.
Instead, it maintains the belief $\mathbf{B}: \mathbf{M}^{env} \to [0, 1]$, 
which is defined as a probability distribution over all possible environment models such that
$\sum_{1 \leq i \leq N} \mathbf{B}(\M^{env}_i) = 1$.
$\mathbf{B}(\M^{env}_i)$ returns the probability that $\M^{env}_i$ is the model being executed by the environment.
The set of all the beliefs forms the belief space, which we denote by $\mathbb{B}$.
In order to obtain the belief at each time step, the system is given the initial belief
$\mathbf{B}_{init} : \M^{env} \to [0, 1] \in \mathbb{B}$.
Then, it subsequently updates the belief using a given belief update function
$\tau : \mathbb{B} \times S^{env} \times S^{env} \to \mathbb{B}$ such that
$\tau(\mathbf{B}, s, s')$ returns the belief after the environment makes a transition from state
$s$ with belief $\mathbf{B}$ to state $s'$.
The belief update function can be defined based on the observation function 
as in the belief MDP construction for POMDPs  \cite{Kaelbling:POMDP1998}.

In general, the belief space $\mathbb{B}$ may be infinite, rendering the control policy synthesis 
computationally intractable to solve exactly.
To overcome this difficulty, we employ techniques for solving POMDPs and
approximate $\mathbb{B}$ by a finite set of representative points from $\mathbb{B}$
and work with this approximate representation instead.
Belief space approximation is beyond the scope of this paper and is subject to future work.
Sampling techniques that have been proposed in the POMDP literature can be found in, e.g., 
\cite{Hauskrecht:VAP2000,Pineau:PVI2003,Kurniawati:SARSOP2008}.

Given a system model described by $\M^{pl}$, $\mathbf{M}^{env} =  \{\M^{env}_1, \M^{env}_2, \ldots, \M^{env}_N\}$,
the (finite) belief space $\mathbb{B}$, the initial belief $\mathbf{B}_{init}$, the belief update function $\tau$
and an LTL formula $\varphi$ that describes the desired property of the system,
we consider the following control policy synthesis problems.

\begin{problem}
\label{prob:expected}
Synthesize a control policy for the system that maximizes the expected 
probability that the system satisfies $\varphi$ where the expected probability
that the environment transitions from state $s \in S^{env}$ with belief $\mathbf{B} \in \mathbb{B}$ to state $s' \in S^{env}$ is given by
$\sum_{1 \leq i \leq N} \mathbf{B}(\M^{env}_i)\mathbf{P}^{env}_i(s, s')$.
\end{problem}

\begin{problem}
\label{prob:worst-case}
Synthesize a control policy for the system that maximizes the worst-case (among all the possible sequences of environment modes) 
probability that the system satisfies $\varphi$.
The environment mode may change during an execution:
when the environment is in state $s \in S^{env}$ with belief $\mathbf{B} \in \mathbb{B}$,
it may switch to any mode $\M^{env}_i \in \mathbf{M}^{env}$
with $\mathbf{B}(\M^{env}_i) > 0$.
We consider both the case where the controller and the environment plays a sequential game
and the case where the control policy and the sequence of environment modes are computed
before an execution.
\end{problem}

\begin{example}
\label{ex:ped-models}
Consider a problem where an autonomous vehicle needs to navigate a road
with a pedestrian walking on the pavement.
The vehicle and the pedestrian are considered the plant and the environment, respectively.
The pedestrian may or may not cross the road, depending on his/her destination, which is unknown to the system.
Suppose the road is discretized into a finite number of cells $c_0, c_2, \ldots, c_M$.
The vehicle is modeled by an MDP 
$\M^{pl} = (S^{pl}, Act, \mathbf{P}^{pl}, s_{init}^{pl}, \Pi^{pl}, L^{pl})$
whose state $s \in S^{pl}$ describes the cell occupied by the vehicle 
and whose action $\alpha \in Act$ corresponds to a motion primitive of the vehicle (e.g., cruise, accelerate, decelerate).
The motion of the pedestrian is modeled by an MC $\M^{env}_i \in \mathbf{M}^{env}$ where
$\mathbf{M}^{env} =  \{\M^{env}_1, \M^{env}_2\}$.
$\M^{env}_1 =  (S^{env}, \mathbf{P}^{env}_1, s_{init}^{env}, \Pi^{env}, L^{env})$
represents the model of the pedestrian if s/he decides not to cross the road whereas
$\M^{env}_2 =  (S^{env}, \mathbf{P}^{env}_2, s_{init}^{env}, \Pi^{env}, L^{env})$
represents the model of the pedestrian if s/he decides to cross the road.
A state $s \in S^{env}$ describes the cell occupied by the pedestrian.
The labeling functions $L^{pl}$ and $L^{env}$ essentially maps each cell to its label,
with an index that identifies the vehicle from the pedestrian,
i.e., $L^{pl}(c_j) = c_{j}^{pl}$ and $L^{env}(c_j) = c_{j}^{env}$ for all $j \in \{0, \ldots M\}$.
Consider the desired property stating that the vehicle does not collide with the pedestrian until
it reaches cell $c_M$ (e.g., the end of the road).
In this case, the specification $\varphi$ can be written as
$\varphi = \left(\neg \bigvee_{j \geq 0} (c^{pl}_j \aand c^{env}_j)\right) \until c^{pl}_M$.
\end{example}

\section{Expectation-Based Control Policy Synthesis}
\label{sec:syn-MDP}
To solve Problem \ref{prob:expected}, we first construct the MDP that represents the complete system, 
taking into account the uncertainties, captured by the belief, in the environment model.
Then, we employ existing results in probabilistic verification and construct the product MDP
and extract its optimal control policy.
In this section, we describe these steps in more detail and discuss their connection to Problem \ref{prob:expected}.

\subsection{Construction of the Complete System}
\label{ssec:comp}
Based on the notion of belief,
we construct the complete environment model, represented by the MC
$\M^{env} = (S^{env} \times \mathbb{B}, \mathbf{P}^{env}, \langle s_{init}^{env}, \mathbf{B}_{init} \rangle, \Pi^{env}, L^{env'})$
where for each $s, s' \in S^{env}$ and $\mathbf{B}, \mathbf{B}' \in \mathbb{B}$, 
	\begin{equation}
		\hspace{-2mm}
		\mathbf{P}^{env}(\langle s, \mathbf{B} \rangle, \langle s', \mathbf{B}' \rangle) = 
		\left\{\begin{array}{ll}
		\vspace{-2mm}
		\displaystyle{\sum_{i} \mathbf{B}(\M^{env}_i)\mathbf{P}^{env}_i(s, s')}\\
		&\hspace{-25mm}\hbox{if } \tau(\mathbf{B}, s, s') = \mathbf{B}'\\
		0 &\hspace{-25mm}\hbox{otherwise}
		\end{array}\right.,
	\end{equation}
and $L^{env'}(s, \mathbf{B}) = L^{env}(s)$.

It is straightforward to check that for all $\langle s, \mathbf{B} \rangle \in S^{env} \times \mathbb{B}$,
$\sum_{s', \mathbf{B}'}  \mathbf{P}^{env}(\langle s, \mathbf{B} \rangle, \langle s', \mathbf{B}' \rangle) = 1$.
Hence, $\M^{env}$ is a valid MC.

Assuming that the plant and the environment make a transition simultaneously,
we obtain the complete system by constructing the synchronous parallel composition of the plant and the environment.
Synchronous parallel composition of MDP and MC is defined as follows.

\begin{definition}
	Let $\M_1 = (S_1, Act, \mathbf{P}_1, s_{init,1}, \Pi_1, L_1)$ be a Markov decision process and
	$\M_2 = (S_2, \mathbf{P}_2, s_{init,2}, \Pi_2, L_2)$ be a Markov chain.
	Their synchronous parallel composition,
	denoted by $\M_1 || \M_2$, 
	is the MDP $\M = (S_1 \times S_2, Act, \mathbf{P}, \langle s_{init,1}, s_{init,2} \rangle, \Pi_1 \cup \Pi_2, L)$ where:
	\begin{itemize}
	\item For each $s_1, s_1' \in S_1$, $s_2, s_2' \in S_2$ and $\alpha \in Act$,
	$\mathbf{P}(\langle s_1, s_2 \rangle, \alpha, \langle s_1', s_2' \rangle) = \mathbf{P}_1(s_1, \alpha, s_1')\mathbf{P}_2(s_2, s_2')$.
	\item For each $s_1 \in S_1$ and $s_2 \in S_2$, $L(\langle s_1, s_2 \rangle) = L(s_1) \cup L(s_2)$.
	\end{itemize}
\end{definition}

From the above definitions, our complete system can be modeled by the MDP $\M^{pl} || \M^{env}$.
We denote this MDP by $\M = (S, Act, \mathbf{P}, s_{init}, \Pi, L)$.
Note that a state $s \in S$ is of the form $s = \langle s^{pl}, s^{env}, \mathbf{B} \rangle$
where $s^{pl} \in S^{pl}$, $s^{env} \in S^{env}$ and $\mathbf{B} \in \mathbb{B}$.
The following lemma shows that Problem \ref{prob:expected} can be solved by finding a control policy $\C$ for $\M$
that maximizes $\prob_{\M}^{\C}(\varphi)$.

\begin{lemma}
Let $r_{\M}^{\C} = s_0s_1\ldots s_n$ be a finite path of $\M$ under policy $\C$
where for each $i$, $s_i = \langle s^{pl}_i ,s^{env}_i, \mathbf{B}_i \rangle \in S^{pl} \times S^{env} \times \mathbb{B}$.
Then,
\begin{equation}
\begin{array}{l}
\prob_{\M}^{\C}\{Path_{\M}^{\C}(r_{\M}^{\C})\} = \\
\hspace{10mm}
\prod_{0 \leq j < n} \Big(\mathbf{P}(s^{pl}_j, \C(s_0s_1\ldots s_j), s^{pl}_{j+1}) \\
\hspace{10mm}
\sum_{1 \leq i \leq N} \mathbf{B}_j(\M^{env}_i)\mathbf{P}^{env}_i(s^{env}_j, s^{env}_{j+1}) \Big)
\end{array}
\end{equation}
Hence, given an LTL formula $\varphi$, $\prob_{\M}^{\C}(\varphi)$ gives
the expected probability that the system satisfies $\varphi$ under policy $\C$.
\end{lemma}
\begin{proof}
The proof straightforwardly follows from the definition of $\M^{env}$ and $\M$.
\end{proof}

\subsection{Construction of the Product MDP}
\label{ssec:pMDP}
Let $\A_\varphi = (Q, 2^\Pi, \delta, q_{init}, Acc)$ be a DRA that recognizes the specification $\varphi$.
Our next step is to obtain a finite MDP $\M_p = (S_p, Act_p, \mathbf{P}_p, s_{p,init}, \Pi_p, L_p)$ as the product of $\M$ and $\A_\varphi$,
defined as follows.

\begin{definition} 
	Let $\M = (S, Act, \mathbf{P}, s_{init}, \Pi, L)$ be an MDP and
	let $\A = (Q, 2^\Pi, \delta, q_{init}, Acc)$ be a DRA.
	Then, the product of $\M$ and $\A$ is the MDP $\M_p = \M \otimes \A$ defined by
	$\M_p = (S_p, Act, \mathbf{P}_p, s_{p,init}, \Pi_p, L_p)$
	where $S_p = S \times Q$, 
	$s_{p,init} = \langle s_{init}, \delta(q_{init}, L(s_{init}) \rangle$,
	$\Pi_p = Q$,
	 $L_p(\langle s,q \rangle)= \{q\}$
	and
	\begin{equation}
	\hspace{-1mm}
	\begin{array}{rcl}
		\mathbf{P}_p( \langle s,q \rangle, \alpha, \langle s',q' \rangle) &\hspace{-2mm}=&\hspace{-2mm} 
		\left\{ \begin{array}{ll} \mathbf{P}(s, \alpha,s') &\hbox{if } q' = \delta(q, L(s'))\\
		0 &\hbox{otherwise} \end{array}\right.
	\end{array}
	\end{equation}
\end{definition}

Consider a path
$r_{\M_p}^{\C_p} = \langle s_0,q_0 \rangle \langle s_1,q_1 \rangle \ldots$ of $\M_p$ under some control policy $\C_p$.
We say that $r_{\M_p}^{\C_p}$ is accepting
if and only if there exists a pair $(H, K) \in Acc$ such that the word generated by $r_{\M_p}^{\C_p}$
intersects with $H$ finitely many times and intersects with $K$ infinitely many times, i.e.,
(1) there exists $n \geq 0$ such that for all $m \geq n$, $L_p(\langle s_m,q_m \rangle) \cap H = \emptyset$, and
(2) there exists infinitely many $n \geq 0$ such that $L_p (\langle s_n,q_n \rangle) \cap K \not= \emptyset$.

Stepping through the above definition shows that 
given a path $r_{\M_p}^{\C_p} = \langle s_0,q_0 \rangle \langle s_1,q_1 \rangle \ldots$ of $\M_p$
generated under some control policy $\C_p$, 
the corresponding path $s_0s_1\ldots$ on $\M$ generates a word
$L(s_0) L(s_1) \ldots$ that satisfies $\varphi$ if and only if $r_{\M_p}^{\C_p}$ is accepting.
Therefore, each accepting path of $\M_p$ uniquely corresponds to a path of $\M$
whose word satisfies $\varphi$.
In addition, a control policy $\C_p$ on $\M_p$ induces a corresponding control policy
$\C$ on $\M$.
The details for generating $\C$ from $\C_p$ can be found, e.g., in \cite{Baier:PMC2008,Ding:LTL2011}.

\subsection{Control Policy Synthesis for Product MDP}
\label{ssec:policy}

From probabilistic verification, it has been shown that 
the maximum probability for $\M$ to satisfy $\varphi$ is 
equivalent to the maximum probability of reaching a certain set of states of $\M_p$
known as \emph{accepting maximal end components} (AMECs).
An \emph{end component} of the product MDP
$\M_p = (S_p, Act_p, \mathbf{P}_p, s_{p,init}, \Pi_p, L_p)$ is a pair
$(T,A)$ where $\emptyset \not= T \subseteq S_p$ and $A : T \to 2^{Act_p}$ such that
(1) $\emptyset \not= A(s) \subseteq Act_p(s)$ for all $s \in T$,
(2) the directed graph induced by $(T,A)$ is strongly connected, and
(3) for all $s \in T$ and $\alpha \in A(s)$, $\{t \in S_p \hspace{1mm}|\hspace{1mm} \mathbf{P}_p(s, \alpha, t) > 0\} \subseteq T$.
An \emph{accepting maximal end component} of $\M_p$ is an end component $(T,A)$ such that
for some $(H,K) \in Acc$, $H \cap T = \emptyset$ and $K \cap T \not= \emptyset$ and
there is no end component $(T', A') \not= (T,A)$
such that $T \subseteq T'$ and $A(s) \subseteq A'(s)$ for all $s \in T$.
It has an important property that starting from any state in $T$,
there exists a finite memory control policy to keep the state within $T$ forever while visiting all states in $T$ infinitely often
with probability 1.
AMECs of $\M_p$ can be efficiently identified based on iterative computations of strongly
connected components of $\M_p$.
We refer the reader to \cite{Baier:PMC2008} for more details.

Once the AMECs of $\M_p$ are identified,
we then compute the maximum probability of reaching $S_G$ where
$S_G$ contains all the states in the AMECs of $\M_p$.
For the rest of the paper, we use an LTL-like notations to describe events
in MDPs.
In particular, we use $\eventually S_G$ to denote the event of reaching some state in $S_G$ eventually.

For each $s \in S_p$, let $x_s$ denote the maximum probability of reaching a state in $S_G$,
starting from $s$.
Formall, $x_s = \sup_{\C_p} \prob^{\C_p}_{\M_p}(s \models \eventually S_G)$.
There are two main techniques for computing the probability $x_s$ for each $s \in S_p$:
linear programming (LP) and value iteration.
LP-based techniques yield an exact solution but
it typically does not scale as well as value iteration.
On the other hand, value iteration is an iterative numerical technique.
This method works by successively computing the probability vector $(x_s^{(k)})_{s \in S_p}$ for increasing $k \geq 0$
such that $\lim_{k \to \infty} x_s^{(k)} = x_s$ for all $s \in S_p$.
Initially, we set $x_s^{(0)} = 1$ if $s \in S_G$ and $x_s^{(0)} = 0$ otherwise.
In the $(k+1)$th iteration where $k \geq 0$, we set
\begin{equation}
\label{eq:value_iteration}
  x_s^{(k+1)} = \left\{ \begin{array}{ll}
  1 &\hbox{if } s \in S_G\\
  \displaystyle{\max_{\alpha \in Act_p(s)}} \sum_{t \in S_p} \mathbf{P}_p(s, \alpha, t)x_t^{(k)} 
  &\hbox{otherwise}.
  \end{array}\right.
\end{equation}

In practice, we terminate the computation and say that $x_s^{(k)}$ converges when a termination criterion such as
$\max_{s \in S_p}|x_s^{(k+1)} - x_s^{(k)}| < \epsilon$ is satisfied for some fixed (typically very small) threshold $\epsilon$.

Once the vector $(x_s)_{s \in S_p}$ is computed, a finite memory control policy $\C_p$ for $\M_p$
that maximizes the probability for $\M$ to satisfy $\varphi$ can be constructed as follows.
First, consider the case when $\M_p$ is in state $s \in S_G$. 
In this case, $s$ belongs to some AMEC $(T, A)$ and
the policy $\C_p$ selects an action $\alpha \in A(s)$ such that
all actions in $A(s)$ are scheduled infinitely often.
(For example, $\C_p$ may select the action for $s$ according to a round-robin policy.)
Next, consider the case when $\M_p$ is in state $s \in S_p \setminus S_G$.
In this case, $\C_p$ picks an action to ensure that $\prob^{\C_p}_\M(s \models \eventually S_G) = x_s$ can be achieved.
If $x_s = 0$, an action in $Act_p(s)$ can be chosen arbitrarily.
Otherwise, $\C_p$ picks an action $\alpha \in Act_p^{max}(s)$
such that $\mathbf{P}_p(s, \alpha, t) > 0$ for some $t \in S_p$ with $\|t\| = \|s\|-1$.
Here, $Act_p^{max}(s) \subseteq Act_p(s)$ is the set of actions such that
for all $\alpha \in Act_p^{max}(s)$,
$x_{s} = \sum_{t \in S_p} \mathbf{P}(s,\alpha,t)x_t$ and
$\|s\|$ denotes the length of a shortest path from $s$ to a state in $S_G$, 
using only actions in $Act_p^{max}$. 

\section{Worst-Case-Based Control Policy Synthesis}
\label{sec:syn-AMDP}

To solve Problem \ref{prob:worst-case}, we first propose a mathematical object called \emph{adversarial Markov decision process} (AMDP).
Then, we show that Problem \ref{prob:worst-case} can be formulated as finding an optimal control policy for an AMDP.
Finally, control policy synthesis for AMDP is discussed.

\subsection{Adversarial Markov Decision Process}
\label{ssec:AMDP}
\begin{definition}
An \emph{adversarial Markov decision process} (AMDP) is a tuple
$\M^{\A} = (S, Act_C, Act_A, \mathbf{P}, s_{init}, \Pi, L)$ where
$S$, $s_{init}$, $\Pi$ and $L$ are defined as in MDP and
\begin{itemize}
\item $Act_C$ is a finite set of control actions,
\item $Act_A$ is a finite set of adversarial actions, and
\item $\mathbf{P} : S \times Act_C \times Act_A \times S \to [0,1]$ is the transition probability function such that
for any $s \in S$, $\alpha \in Act_C$ and $\beta \in Act_A$,
$\sum_{t \in S} \mathbf{P}(s, \alpha, \beta, t) \in \{0, 1\}$.
\end{itemize}
\end{definition}
We say that a control action $\alpha$ is \emph{enabled} in state $s$ if and only if
there exists an adversarial action $\beta$ such that
$\sum_{t \in S} \mathbf{P}(s, \alpha, \beta, t) = 1$.
Similarly, an adversarial action $\beta$ is \emph{enabled} in state $s$ if and only if
there exists a control action $\alpha$ such that 
$\sum_{t \in S} \mathbf{P}(s, \alpha, \beta, t) = 1$.
Let $Act_C(s)$ and $Act_A(s)$ denote the set of enabled control and adversarial actions, respectively, in $s$.
We assume that for all $s \in S$, $\alpha \in Act_C(s)$ and $\beta \in Act_A(s)$,
$\sum_{t \in S} \mathbf{P}(s, \alpha, \beta, t) = 1$,
i.e., whether an adversarial (resp. control) action is enabled in state $s$
depends only on the state $s$ itself but
not on a control (resp. adversarial) action taken by the system (resp. adversary).

Given an AMDP $\M^{\A} = (S, Act_C, Act_A, \mathbf{P}, s_{init}, \Pi, L)$,
a control policy $\C : S^+ \to Act_C$ and an adversarial policy $\D : S^+ \to Act_A$
for an AMDP can be defined such that $\C(s_0s_1\ldots s_n) \in Act_C(s_n)$
and $\D(s_0s_1\ldots s_n) \in Act_A(s_n)$ for all $s_0s_1\ldots s_n \in S^+$.
A unique policy measure $\prob^{\C, \D}_{\M^{\A}}$ on
the $\sigma$-algebra associated with $\M^{\A}$ under control policy $\C$ and
adversarial policy $\D$ can then be defined based on the notion of path on $\M^{\A}$
as for an ordinary MDP.

We end the section with important properties of AMDP that will be employed in the control policy synthesis.

\begin{definition}
\label{def:complete}
Let $\mathbb{G}$ be a set of functions from $\T$ to $\V$ where
$\T$ and $\V$ are finite sets.
We say that $\mathbb{G}$ is \emph{complete} if
for any $t_1, t_2 \in \T$ and
$g_1, g_2 \in \mathbb{G}$,
there exists $g \in \mathbb{G}$ such that
$g(t_1) = g_1(t_1)$ and $g(t_2) = g_2(t_2)$ .
\end{definition}

\begin{lemma}
\label{lem:minmax-duality}
Let $\T$, $\U$ and $\V$ be finite sets
and let $\mathbb{G}$ be a set of functions from $\T$ to $\V$.
Then, $\mathbb{G}$ is finite.
Furthermore, suppose that $\mathbb{G}$ is complete.
Then, for any $F : \U \times \T \to \nnreals$
and $G : \V \times \T \to \nnreals$,
\begin{equation}
\label{eq:minmax-duality}
\begin{array}{l}
\min_{u \in \U} \sum_{t \in \T} F(u, t) \max_{g \in \mathbb{G}} G(g(t), t)\\
\hspace{5mm}= \min_{u \in \U} \max_{g \in \mathbb{G}} \sum_{t \in \T} F(u, t) G(g(t), t)\\
\hspace{5mm}= \max_{g \in \mathbb{G}}  \min_{u \in \U}\sum_{t \in \T} F(u, t) G(g(t), t).
\end{array}
\end{equation}
In addition,
\begin{equation}
\label{eq:minmin-duality}
\begin{array}{l}
\min_{u \in \U} \sum_{t \in \T} F(u, t) \min_{g \in \mathbb{G}} G(g(t), t)\\
\hspace{5mm}= \min_{u \in \U} \min_{g \in \mathbb{G}} \sum_{t \in \T} F(u, t) G(g(t), t).
\end{array}
\end{equation}
\end{lemma}
\begin{proof}
Since both $\T$ and $\V$ are finite, clearly, $\mathbb{G}$ is finite.
Thus, the $\min$ and $\max$ in (\ref{eq:minmax-duality})--(\ref{eq:minmin-duality})
are well defined.
Let $\tilde{g} \in \mathbb{G}$ be a function such that
$\sum_{t \in \T} F(u, t) G(\tilde{g}(t), t) =  \max_{g \in \mathbb{G}} \sum_{t \in \T} F(u, t) G(g(t), t)$.
Since $\mathbb{G}$ is complete, there exists a function $g^* \in \mathbb{G}$ such that for all $t \in \T$, $g^*(t)$ satisfies
$G(g^*(t), t) = \max_{g \in \mathbb{G}} G(g(t), t)$.
Furthermore, since both $F$ and $G$ are non-negative and $\tilde{g} \in \mathbb{G}$, it follows that
$F(u, t) G(g^*(t), t) = \max_{g \in \mathbb{G}} F(u, t) G(g(t), t) \geq F(u, t) G(\tilde{g}(t), t)$
for all $u \in \U$ and $t \in \T$.
Taking the sum over all $t \in \T$, we get
\begin{equation}
\label{eq:pf-minmax-duality1}
\begin{array}{rcl}
\sum_{t \in \T}  F(u, t) G(g^*(t), t) 
&\geq&
\sum_{t \in \T} F(u, t) G(\tilde{g}(t), t) \\
&=&
\displaystyle{\max_{g \in \mathbb{G}} \sum_{t \in \T} F(u, t) G(g(t), t)}
\end{array}
\end{equation}
But since $g^* \in \mathbb{G}$, it follows that
\begin{equation}
\label{eq:pf-minmax-duality2}
\max_{g \in \mathbb{G}} \sum_{t \in \T} F(u, t) G(g(t), t)
\geq \sum_{t \in \T} F(u, t) G(g^*(t), t).
\end{equation}

Combining (\ref{eq:pf-minmax-duality1}) and (\ref{eq:pf-minmax-duality2}), 
we get that all the inequalities must be replaced by equalities.
Hence, we can conclude that the first equality in (\ref{eq:minmax-duality}) holds.
The proof for the equality in (\ref{eq:minmin-duality}) follows similar arguments.
Hence, we only provide a proof for the second inequality in (\ref{eq:minmax-duality}).

First, from weak duality, we know that
\begin{equation}
\label{eq:pf-minmax-duality3}
\begin{array}{l}
\min_{u \in \U} \max_{g \in \mathbb{G}} \sum_{t \in \T} F(u, t) G(g(t), t)\\
\hspace{5mm}\geq \max_{g \in \mathbb{G}}  \min_{u \in \U}\sum_{t \in \T} F(u, t) G(g(t), t).
\end{array}
\end{equation}
For each $g \in \mathbb{G}$, consider an element $u_g^* \in \U$ such that
$\sum_{t \in \T} F(u_g^*, t) G(g(t), t) = \min_{u \in \U}\sum_{t \in \T} F(u, t) G(g(t), t)$.
Since $g^* \in \mathbb{G}$, it follows that
\begin{equation}
\label{eq:pf-minmax-duality4}
\max_{g \in \mathbb{G}}  \min_{u \in \U}\sum_{t \in \T} F(u, t) G(g(t), t) \geq
\sum_{t \in \T} F(u_{g^*}^*, t) G(g^*(t), t).
\end{equation}
But, from the definition of $g^*$ and $u^*$ and the first equality in (\ref{eq:minmax-duality}), we also get
\begin{equation}
\label{eq:pf-minmax-duality5}
\begin{array}{l}
\sum_{t \in \T} F(u_{g^*}^*, t) G(g^*(t), t) \\
\hspace{5mm}=
\min_{u \in \U} \sum_{t \in \T} F(u, t) \max_{g \in \mathbb{G}} G(g(t), t)\\
\hspace{5mm}=
\min_{u \in \U} \max_{g \in \mathbb{G}} \sum_{t \in \T} F(u, t) G(g(t), t).
\end{array}
\end{equation}
Following the chain of inequalities in (\ref{eq:pf-minmax-duality3})--(\ref{eq:pf-minmax-duality5}),
we can conclude that the second equality in (\ref{eq:minmax-duality}) holds.
\end{proof}

\begin{proposition}
\label{prop:xs-AMDP-alt}
Let $\M^{\A} = (S, Act_C, Act_A, \mathbf{P}, s_{init}, \Pi$, $L)$ be a finite AMDP and
$S_G \subseteq S$ be the set of goal states.
Let
\begin{equation}
\label{eq:xs-AMDP-alt}
x_s = \sup_{\C_0 \in \mathbb{C}_0} \inf_{\D_0 \in \mathbb{D}_0} \sup_{\C_1 \in \mathbb{C}_1} \inf_{\D_1 \in \mathbb{D}_1} \ldots 
\prob^{\C,\D}_{\M^{\A}}(s \models \eventually S_G),
\end{equation}
where for any $n \geq 0$,
$\mathbb{C}_n = \{\C : S^{n+1} \to Act_C \hspace{1mm}|\hspace{1mm} \C(s_0s_1\ldots s_n) \in Act_C(s_n)\}$,
$\mathbb{D}_n = \{\D : S^{n+1} \to Act_A \hspace{1mm}|\hspace{1mm} \D(s_0s_1\ldots s_n) \in Act_A(s_n)\}$,
$\C(s_0s_1\ldots s_n) = \C_n(s_0s_1\ldots s_n)$ and
$\D(s_0s_1\ldots s_n) = \D_n(s_0s_1\ldots s_n)$.
For each $k \geq 0 $, consider a vector $(x_s^{(k)})_{s \in S}$ where
$x_s^{(0)} = 1$ for all $s \in S_G$,
$x_s^{(0)} = 0$ for all $s \not\in S_G$
and for all $k \geq 0$,
\begin{equation}
\label{eq:xsk-AMDP-alt}
x_s^{(k+1)} = \left\{ \begin{array}{ll}
1 &\hspace{-10mm}\hbox{if } s \in S_G\\
\vspace{-3mm}\displaystyle{\max_{ \alpha \in Act_C(s)} \min_{\beta \in Act_A(s)}} \sum_{t \in S} \mathbf{P}(s, \alpha, \beta, t) x_{t}^{(k)}\\
&\hspace{-10mm}\hbox{otherwise}
\end{array}\right.
\end{equation}
Then, for any $s \in S$, $x_s^{(0)} \leq x_s^{(1)} \leq \ldots \leq x_s$
and $x_s = \lim_{k \to \infty} x_s^{(k)}$.
\end{proposition}
\begin{proof}
Since $\sum_{t \in S} \mathbf{P}(s, \alpha, \beta, t) \in \{0,1\}$ for all $\alpha \in Act_C$ and $\beta \in Act_A$,
it can be checked that for any $k \geq 0$, if $x_s^{(k)} \in [0,1]$ for all $s \in S$,
then $x_s^{(k+1)} \in [0,1]$ for all $s \in S$.
Since, $x_s^{(0)} \in [0,1]$ for all $s \in S$,
we can conclude that $x_s^{(k)} \in [0,1]$ for all $k \geq 0$ and $s \in S$.

Let $\eventually^{\leq k} S_G$ denote the event of reaching some state in $S_G$
within $k$ steps.
We will show, using induction on $k$, that for any $k \geq 0$ and $s \in S$,
\begin{equation}
\label{eq:pf-alt-policies}
\hspace{-2mm}
\begin{array}{rcl}
x_s^{(k)} &\hspace{-1mm}=&\hspace{-1mm} 
\sup_{\C_0 \in \mathbb{C}_0} \inf_{\D_0 \in \mathbb{D}_0} \sup_{\C_1 \in \mathbb{C}_1} \inf_{\D_1 \in \mathbb{D}_1} \ldots \\
&&\hspace{-1mm}
\sup_{\C_{k-1} \in \mathbb{C}_{k-1}} \inf_{\D_{k-1} \in \mathbb{D}_{k-1}}
\prob^{\C, \D}_{\M^{\A}} (s \models \eventually^{\leq k} S_G),
\end{array}
\end{equation}
where $\C$ and $\D$ are control and adversarial policies such that
for any $n$ such that $0 \leq n < k$,
$\C(s_0s_1\ldots s_n) = \C_n(s_0s_1\ldots s_n)$ and
$\D(s_0s_1\ldots s_n) = \D_n(s_0s_1\ldots s_n)$.
The case where $s \in S_G$ is trivial so we only consider an arbitrary $s \in S \setminus S_G$.
Clearly, $x_s^{(0)} = \prob^{\C, \D}_{\M^{\A}} (s \models \eventually^{\leq 0} S_G)$
for any control and adversarial policies $\C$ and $\D$.
Consider an arbitrary $k \geq 0$ and assume that for all $s \in S$,
(\ref{eq:pf-alt-policies}) holds.
Then, from (\ref{eq:xsk-AMDP-alt}), we get that for any $s \in S \setminus S_G$,
\begin{equation*}
\hspace{-2mm}
\begin{array}{rcl}
x_s^{(k+1)} &\hspace{-2mm}=&\hspace{-2mm} \max_{ \alpha \in Act_C(s)} \min_{\beta \in Act_A(s)} \sum_{t \in S} \mathbf{P}(s, \alpha, \beta, t)\\
&&\hspace{-2mm} 
\sup_{\C_0 \in \mathbb{C}_0} \inf_{\D_0 \in \mathbb{D}_0}  \sup_{\C_1 \in \mathbb{C}_1} \inf_{\D_1 \in \mathbb{D}_1}\ldots \\
&&\hspace{-2mm}
\sup_{\C_{k-1} \in \mathbb{C}_{k-1}} \inf_{\D_{k-1} \in \mathbb{D}_{k-1}} \prob^{\C, \D}_{\M^{\A}} (t \models \eventually^{\leq k} S_G).
\end{array}
\end{equation*}
Since for all $n$ such that $0 \leq n < k$, $S^{n+1}$, $Act_C$ and $Act_A$ are finite,
it follows that $\mathbb{C}_n$ and $\mathbb{D}_n$ are finite.
Furthermore, $\prob^{\C, \D}_{\M^{\A}} (t \models \eventually^{\leq k} S_G)$ only depends
on $\C_0$, $\D_0, \ldots, \C_{k-1}$, $\D_{k-1}$.
Thus, we can conclude that all the $\sup$ and $\inf$ above can be attained,
so we can replace them by $\max$ and $\min$, respectively.

Consider arbitrary $s \in S \setminus S_G$ and $\alpha \in Act_C(s)$.
Define a function $F : Act_A(s) \times S \to [0,1]$ such that
$F(\beta, t) = \mathbf{P}(s, \alpha, \beta, t)$.
In addition, define a set $\mathbb{G} = \mathbb{C}_0$ and a function
$G : Act_C \times S \to [0,1]$ such that
$G(\C_0(t), t) = \inf_{\D_0 \in \mathbb{D}_0}  \sup_{\C_1 \in \mathbb{C}_1} \inf_{\D_1 \in \mathbb{D}_1}\ldots 
\sup_{\C_{k-1} \in \mathbb{C}_{k-1}} \inf_{\D_{k-1} \in \mathbb{D}_{k-1}}$ 
$\prob^{\C, \D}_{\M^{\A}} (t \models \eventually^{\leq k} S_G)$
where for any $n \in \{0, \ldots, k-1\}$,
$\C(s_0s_1\ldots s_n) = \C_n(s_0s_1\ldots s_n)$ and
$\D(s_0s_1\ldots s_n) = \D_n(s_0s_1\ldots s_n)$.
Pick arbitrary $t_1, t_2 \in S$ and $g_1, g_2 \in \mathbb{G}$.
Suppose $g_1(t_1) = \alpha_1$ and $g_2(t_2) = \alpha_2$.
Then, from the definition of $\mathbb{C}_0$,
$\alpha_1 \in Act_C(t_1)$ and $\alpha_2 \in Act_C(t_2)$;
hence, there must exists $g \in \mathbb{G}$ such that $g(t_1) = \alpha_1$
and $g(t_2) = \alpha_2$.
Thus, by definition, $\mathbb{G}$ is complete.
Applying Lemma \ref{lem:minmax-duality}, we get
\begin{equation*}
\hspace{-2mm}
\begin{array}{rcl}
x_s^{(k+1)} &\hspace{-2mm}=&\hspace{-2mm} \max_{ \alpha \in Act_C(s)} \min_{\beta \in Act_A(s)} \max_{\C_0 \in \mathbb{C}_0}\\
&&\hspace{-2mm} 
\sum_{t \in S} \mathbf{P}(s, \alpha, \beta, t) \min_{\D_0 \in \mathbb{D}_0} \ldots \max_{\C_{k-1} \in \mathbb{C}_{k-1}} \\
&&\hspace{-2mm}
 \min_{\D_{k-1} \in \mathbb{D}_{k-1}} \prob^{\C, \D}_{\M^{\A}} (t \models \eventually^{\leq k} S_G).
\end{array}
\end{equation*}

Applying a similar procedure as in the previous paragraph $2k$ times, 
we get 
\begin{equation*}
\hspace{-2mm}
\begin{array}{rcl}
x_s^{(k+1)} &\hspace{-2mm}=&\hspace{-2mm} \displaystyle{\max_{ \alpha \in Act_C(s)} \min_{\beta \in Act_A(s)}} 
\max_{\C_0 \in \mathbb{C}_0}  \min_{\D_0 \in \mathbb{D}_0}\ldots \max_{\C_{k-1} \in \mathbb{C}_{k-1}}\\
&&\hspace{-2mm}
\displaystyle{\min_{\D_{k-1} \in \mathbb{D}_{k-1}} \sum_{t \in S}} \mathbf{P}(s, \alpha, \beta, t)
\prob^{\C, \D}_{\M^{\A}} (t \models \eventually^{\leq k} S_G)\\
&\hspace{-2mm}=&\hspace{-2mm}
\displaystyle{\sup_{\C_0 \in \mathbb{C}_0} \inf_{\D_0 \in \mathbb{D}_0} \sup_{\C_1 \in \mathbb{C}_1} \inf_{\D_1 \in \mathbb{D}_1} \ldots 
\sup_{\C_{k} \in \mathbb{C}_{k}} \inf_{\D_{k} \in \mathbb{D}_{k}}} \\
&&\hspace{-2mm}
\prob^{\C, \D}_{\M^{\A}} (s \models \eventually^{\leq k+1} S_G).
\end{array}
\end{equation*}

Thus, we can conclude that for any $k \geq 0$ and $s \in S \setminus S_G$, (\ref{eq:pf-alt-policies}) holds.
Furthermore, since the set of events $\eventually^{\leq k+1} S_G$ includes the set of events $\eventually^{\leq k} S_G$,
we obtain $x_s^{(k)} \leq x_s^{(k+1)}$.
Finally, we can conclude that $\lim_{k \to \infty} x_s^{(k)} = x_s$ using the fact that 
the sequence $x_s^{(0)}, x_s^{(1)}, \ldots$ is monotonic and bounded (and hence has a finite limit)
and $\eventually S_G$ is the countable union of the events $\eventually^{\leq k} S_G$.
\end{proof}

\begin{proposition}
\label{prop:xs-AMDP-nonalt}
Let $\M^{\A} = (S, Act_C, Act_A, \mathbf{P}, s_{init}, \Pi$, $L)$ be a finite AMDP and
$S_G \subseteq S$ be the set of goal states.
Let
\begin{equation}
\label{eq:xs-AMDP-nonalt}
y_s = \sup_{\C} \inf_{\D} \prob^{\C, \D}_{\M^{\A}} (s \models \eventually S_G).
\end{equation}
For each $k \geq 0 $, consider a vector $(y_s^{(k)})_{s \in S}$ where
$y_s^{(0)} = 1$ for all $s \in S_G$,
$y_s^{(0)} = 0$ for all $s \not\in S_G$
and for all $k \geq 0$,
\begin{equation}
\label{eq:xsk-AMDP-nonalt}
y_s^{(k+1)} = \left\{ \begin{array}{ll}
1 &\hspace{-10mm}\hbox{if } s \in S_G\\
\vspace{-3mm}\displaystyle{\max_{ \alpha \in Act_C(s)} \min_{\beta \in Act_A(s)}} \sum_{t \in S} \mathbf{P}(s, \alpha, \beta, t) y_{t}^{(k)}\\
&\hspace{-10mm}\hbox{otherwise}
\end{array}\right.
\end{equation}
Then, for any $s \in S$, $y_s^{(0)} \leq y_s^{(1)} \leq \ldots \leq y_s$
and $y_s = \lim_{k \to \infty} y_s^{(k)}$.
\end{proposition}
\begin{proof}
The proof closely follows the proof of Proposition \ref{prop:xs-AMDP-alt}.
Roughly, we show, by induction on $k$ and applying Lemma \ref{lem:minmax-duality}, 
that for any $k \geq 0$ and $s \in S$,
$y_s^{(k)} = \sup_{\C} \inf_{\D} \prob^{\C, \D}_{\M^{\A}} (s \models \eventually^{\leq k} S_G)$.
We can conclude the proof using a similar argument as in the proof of Proposition \ref{prop:xs-AMDP-alt}.
\end{proof}

From Proposition \ref{prop:xs-AMDP-alt} and Proposition \ref{prop:xs-AMDP-nonalt},
we can conclude that the sequential game (\ref{eq:xs-AMDP-alt}) is equivalent to its
nonsequential counterpart (\ref{eq:xs-AMDP-nonalt}).

\begin{corollary}
The maximum worst-case probability of reaching a set $S_G$ of states 
in an AMDP $\M^{\A}$ does not depend on whether the controller and the adversary play
alternatively or both the control and adversarial policies are computed at the beginning of an execution. 
\end{corollary}

\subsection{The Complete System as an AMDP}
\label{ssec:AMDP-construction}
We start by constructing an MDP $\M = (S, Act, \mathbf{P}, s_{init}, \Pi, L)$ 
that represents the complete system as described in Section \ref{ssec:comp}.
As discussed earlier, a state of $\M$ is of the form $\langle s^{pl}, s^{env}, \mathbf{B} \rangle$
where $s^{pl} \in S^{pl}$, $s^{env} \in S^{env}$ and $\mathbf{B} \in \mathbb{B}$.
The corresponding AMDP $\M_{\A}$ of $\M$ is then defined as
$\M_{\A} = (S, Act_C, Act_A, \mathbf{P}_{\A}, s_{init}, \Pi, L)$ where
$Act_C = Act$,
$Act_A = \{\beta_1, \ldots, \beta_N\}$ 
(i.e., $\beta_i$ corresponds to the environment choosing model $\M^{env}_i)$ and
for any $s^{pl}_1, s^{pl}_2 \in S^{pl}$, $s^{env}_1, s^{env}_2 \in S^{env}$, 
$\mathbf{B}_1, \mathbf{B}_2 \in \mathbb{B}$,
$\alpha \in Act_C$ and $1 \leq i \leq N$,
\begin{equation}
\label{eq:system-AMDP}
\begin{array}{l}
\mathbf{P}_{\A}(\langle s^{pl}_1, s^{env}_1, \mathbf{B}_1 \rangle, \alpha, \beta_i, \langle s^{pl}_2, s^{env}_2, \mathbf{B}_2 \rangle) = \\
\hspace{5mm}
\left\{ \begin{array}{ll}
\mathbf{P}^{pl}(s^{pl}_1, \alpha, s^{pl}_2)\mathbf{P}^{env}_i(s^{env}_1, s^{env}_2) 
&\hbox{if } \mathbf{B}_1(\M^{env}_i) > 0\\
0 &\hbox{otherwise}
\end{array}\right.,
\end{array}
\end{equation}
if $\tau(\mathbf{B}_1, s, s') = \mathbf{B}_2$;
otherwise, $\mathbf{P}_{\A}(\langle s^{pl}_1, s^{env}_1, \mathbf{B}_1 \rangle, \alpha, \beta_i, \langle s^{pl}_2, s^{env}_2, \mathbf{B}_2 \rangle) = 0$.

It is straightforward to check that $\M_\A$ is a valid AMDP.
Furthermore, based on this construction and the assumptions that 
(1) at any plant state $s^{pl} \in S^{pl}$, 
there exists an action that is enabled in $s^{pl}$, and 
(2) at any point in an execution, the belief $\mathbf{B} \in \mathbb{B}$ satisfies $\sum_{1 \leq i \leq N} \mathbf{B}(\M^{env}_i) = 1$, 
it can be shown that at any state $s \in S$, 
there exists a control action $\alpha \in Act_C$ and an adversarial action $\beta \in Act_A$
that are enabled in $s$.
In addition, consider the case where the environment is in state $s^{env} \in S^{env}$ with belief $\mathbf{B} \in \mathbb{B}$.
It can be shown that for all $\M^{env}_i \in \mathbf{M}^{env}$,
if $\mathbf{B}(\M^{env}_i) > 0$, then $\beta_i$ is enabled in 
$\langle s^{pl}, s^{env}, \mathbf{B} \rangle$ for all $s^{pl} \in S^{pl}$.
Thus, we can conclude that $\M_\A$ represents the complete system for Problem \ref{prob:worst-case}.

\begin{remark}
\label{remark:AMDP-construction}
According to (\ref{eq:system-AMDP}),
the system does not need to maintain the exact belief in each state.
The only information needed to construct an AMDP that represents the complete system
is all the possible modes of the environment in each state of the complete system.
This allows us to integrate methodologies for discrete state estimation \cite{Vecchio06discretestate}
to reduce the size of the AMDP.
This direction is subject to future work.
\end{remark}

\subsection{Control Policy Synthesis for AMDP}

Similar to control policy synthesis for MDP, control policy synthesis for AMDP $\M^{\A}$ can be done
on the basis of a product construction.
The product of $\M^{\A} = (S, Act_C, Act_A, \mathbf{P}, s_{init}, \Pi, L)$ and 
DRA $\A = (Q, 2^\Pi, \delta, q_{init}, Acc)$ is an AMDP 
$\M^{\A}_p = (S_p, Act_C, Act_A, \mathbf{P}_p, s_{p, init}, \Pi_p, L_p)$, which is defined similar
to the product of MDP and DRA, except that the set of actions is partitioned into
the set of control and the set of adversarial actions.

Following the steps for synthesizing a control policy for product MDP,
we identify the AMECs of $\M^{\A}_p$.
An AMEC of $\M^{\A}_p$ is defined
based on the notion of end component as for the case of product MDP.
However, an end component of $\M^{\A}_p$ needs to be defined, 
taking into account the adversary.
Specifically, an end component of $\M^{\A}_p$ is a pair
$(T,A)$ where $\emptyset \not= T \subseteq S_p$ and $A : T \to 2^{Act_C}$ such that
(1) $\emptyset \not= A(s) \subseteq Act_C(s)$ for all $s \in T$,
(2) the directed graph induced by $(T,A)$ under any adversarial policy is strongly connected, and
(3) for all $s \in T$, $\alpha \in A(s)$ and $\beta \in Act_A(s)$, 
$\{t \in S_p \hspace{1mm}|\hspace{1mm} \mathbf{P}_p(s, \alpha, \beta, t) > 0\} \subseteq T$.

Using a similar argument as in the case of product MDP \cite{Baier:PMC2008},
it can be shown that 
the maximum worst-case probability for $\M^{\A}$ to satisfy $\varphi$
is equivalent to the maximum worst-case probability of reaching a states in an AMEC of $\M^{\A}_p$.
We can then apply Proposition \ref{prop:xs-AMDP-alt} and Proposition \ref{prop:xs-AMDP-nonalt}
to compute $x_s$, which is equivalent to $y_s$, using value iteration.
A control policy for $\M^{\A}_p$ that maximizes the worst-case probability
for $\M^{\A}$ to satisfy $\varphi$ can be constructed as outlined at the end of Section \ref{ssec:policy}
for product MDP.

\section{Example}
\label{sec:ex}
Consider, once again, the autonomous vehicle problem described in Example \ref{ex:ped-models}.
Suppose the road is discretized into 9 cells $c_0, \ldots, c_8$ as shown in Figure \ref{fig:example}.
The vehicle starts in cell $c_0$ and has to reach cell $c_8$ whereas
the pedestrian starts in cell $c_1$.
The models of the vehicle and the pedestrian are shown in Figure \ref{fig:models}.
The vehicle has two actions $\alpha_1$ and $\alpha_2$, which correspond to decelerating and accelerating, respectively.
The pedestrian has 2 modes $\M^{env}_1$ and $\M^{env}_2$, which correspond to the cases where
s/he wants to remain on the left side of the road and cross the road, respectively.
A DRA $\A_\varphi$ that accepts all and only words that satisfy
$\varphi = \left(\neg \bigvee_{j \geq 0} (c^{pl}_j \aand c^{env}_j)\right) \until c^{pl}_8$
is shown in Figure \ref{fig:Aphi}.
Finally, we consider the set $\mathbb{B} = \{\mathbf{B}_0, \ldots, \mathbf{B}_8\}$ of beliefs where for all $i$,
$\mathbf{B}_i(\M^{env}_1) = 0.1i$ and $\mathbf{B}_i(\M^{env}_2) = 1-0.1i$.
We set $\mathbf{B}_{init} = \mathbf{B}_6$ where it is equally likely that the pedestrian is
in mode $\M^{env}_1$ or mode $\M^{env}_2$.
The belief update function $\tau$ is defined such that
the longer the pedestrian stay on the left side of the road,
the probability that s/he is in mode $\M^{env}_1$ increases.
Once the pedestrian starts crossing the road, we change the belief to $\mathbf{B}_0$ where
it is certain that the pedestrian is in mode $\M^{env}_2$.
Specifically, for all $i$, we let
\begin{equation}
\label{eq:ex-tau}
\tau(\mathbf{B}_i, s, s') = 
\left\{ \begin{array}{ll}
\mathbf{B}_0 &\hbox{if } s' \in \{c_2, c_4, c_6, c_7\}\\
\mathbf{B}_i &\hbox{if } i = 0 \hbox{ or } i = 8\\
\mathbf{B}_{i+1} &\hbox{otherwise}
\end{array}\right..
\end{equation}

\begin{figure}
\centering\includegraphics[width=0.4\textwidth]{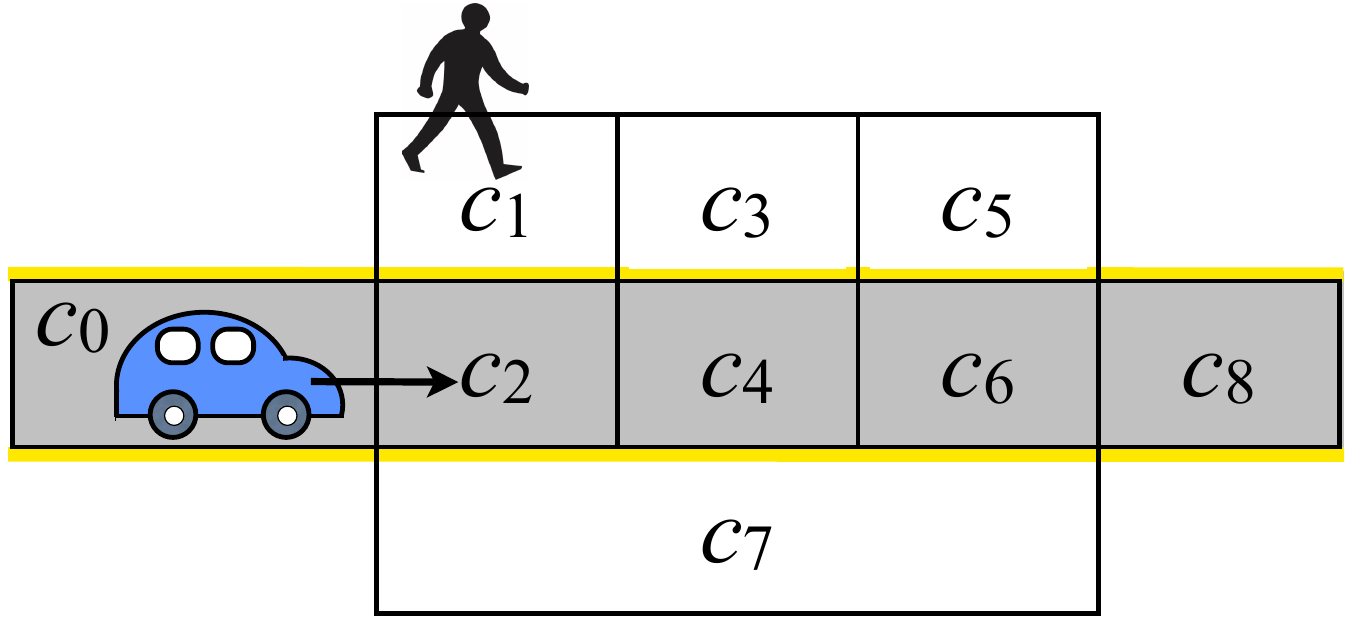}
\caption{The road and its partition used in the autonomous vehicle example.}
\label{fig:example}
\end{figure}

\begin{figure}
\centering 
\subfigure[Vehicle model $\M^{pl}$]{
	\hspace{-10mm}
	\begin{tikzpicture}[->,>=stealth',shorten >=1pt,auto,node distance=2cm]
	  \tikzstyle{every state}=[circle,thick,draw=blue!75,minimum size=6mm]

	   \node [state] (c0) at (0,0) {$c_0$};
	   \node [state] (c2) at (1.8,0){$c_2$};
	   \node [state] (c4) at (3.6,0) {$c_4$};
	   \node [state] (c6) at (5.4,0) {$c_6$};
	   \node [state] (c8) at (7.2,0) {$c_8$};
	   
	   \draw [ shorten >= 1pt, -> ] (-0.9,0) to (c0);
	   
	   \path (c0) edge [loop above] node {$\alpha_1, 0.9$} (c0)
	   		  edge [loop above,in=50,out=130,looseness=12] node {$\alpha_2, 0.1$} (c0)
	                     edge node [above] {$\alpha_1, 0.1$} (c2)
	                     edge [bend right ] node [below] {$\alpha_2, 0.9$} (c2)
	             (c2) edge [loop above] node {$\alpha_1, 0.9$} (c2)
	   		  edge [loop above,in=50,out=130,looseness=12] node {$\alpha_2, 0.1$} (c2)
	                     edge node [above] {$\alpha_1, 0.1$} (c4)
	                     edge [bend right ] node [below] {$\alpha_2, 0.9$} (c4)
	             (c4) edge [loop above] node {$\alpha_1, 0.9$} (c4)
	   		  edge [loop above,in=50,out=130,looseness=12] node {$\alpha_2, 0.1$} (c4)
	                     edge node [above] {$\alpha_1, 0.1$} (c6)
	                     edge [bend right ] node [below] {$\alpha_2, 0.9$} (c6)
	             (c6) edge [loop above] node {$\alpha_1, 0.9$} (c6)
	   		  edge [loop above,in=50,out=130,looseness=12] node {$\alpha_2, 0.1$} (c6)
	                     edge node [above] {$\alpha_1, 0.1$} (c8)
	                     edge [bend right ] node [below] {$\alpha_2, 0.9$} (c8)
	             (c8) edge [loop above] node {$\alpha_1, 1$} (c8);
	\end{tikzpicture}
}

\subfigure[Pedestrian model $\M^{env}_1$]{
	\begin{tikzpicture}[->,>=stealth',shorten >=1pt,auto,node distance=2cm]
	  \tikzstyle{every state}=[circle,thick,draw=blue!75,minimum size=6mm]

	   \node [state] (c1) at (1,0) {$c_1$};
	   \node [state] (c3) at (3,0){$c_3$};
	   \node [state] (c5) at (5,0) {$c_5$};
	   
	    \draw [ shorten >= 1pt, -> ] (0.1,0) to (c1);
	   
	   \path (c1) edge [loop above] node {$0.4$} (c1)
	                     edge [bend left] node {$0.6$} (c3)
	             (c3) edge [loop above] node {$0.3$} (c3)
	                     edge [bend left] node {$0.35$} (c1)
	                     edge [bend left] node {$0.35$} (c5)
	             (c5) edge [loop above] node {$0.7$} (c5)
	                     edge [bend left] node {$0.3$} (c3);
	\end{tikzpicture}
}

\subfigure[Pedestrian model $\M^{env}_2$]{
	\begin{tikzpicture}[->,>=stealth',shorten >=1pt,auto,node distance=2cm ,bend angle=15]
	  \tikzstyle{every state}=[circle,thick,draw=blue!75,minimum size=6mm]

	   \node [state] (c1) at (1,0) {$c_1$};
	   \node [state] (c2) at (4,0.9) {$c_2$};
	   \node [state] (c3) at (3,-0.9) {$c_3$};
	   \node [state] (c4) at (6,0) {$c_4$};
	   \node [state] (c5) at (6,-1.7) {$c_5$};
	   \node [state] (c6) at (8,-1.7) {$c_6$};
	   \node [state] (c7) at (8,0) {$c_7$};
	   
    	   \draw [ shorten >= 1pt, -> ] (0.1,0) to (c1);
	   
	   \path (c1) edge [loop above] node {$0.3$} (c1)
	                     edge [] node {$0.6$} (c2)
	                     edge [bend left] node {$0.1$} (c3)
	             (c2) edge [bend left] node {$1$} (c7)
	             (c3) edge [loop above] node {$0.3$} (c3)
	                     edge [bend left] node {$0.05$} (c1)
	                     edge [] node {$0.6$} (c4)
	                     edge [bend left] node {$0.05$} (c5)
	             (c4) edge [] node {$1$} (c7)
	             (c5) edge [loop above] node {$0.3$} (c5)
	                     edge [] node {$0.6$} (c6)
	                     edge [bend left] node {$0.1$} (c3)
	             (c6) edge [] node {$1$} (c7)
	             (c7) edge [loop above] node {$1$} (c7);
	\end{tikzpicture}
}

   \caption{Vehicle and pedestrian models.}
   \label{fig:models}
\end{figure}
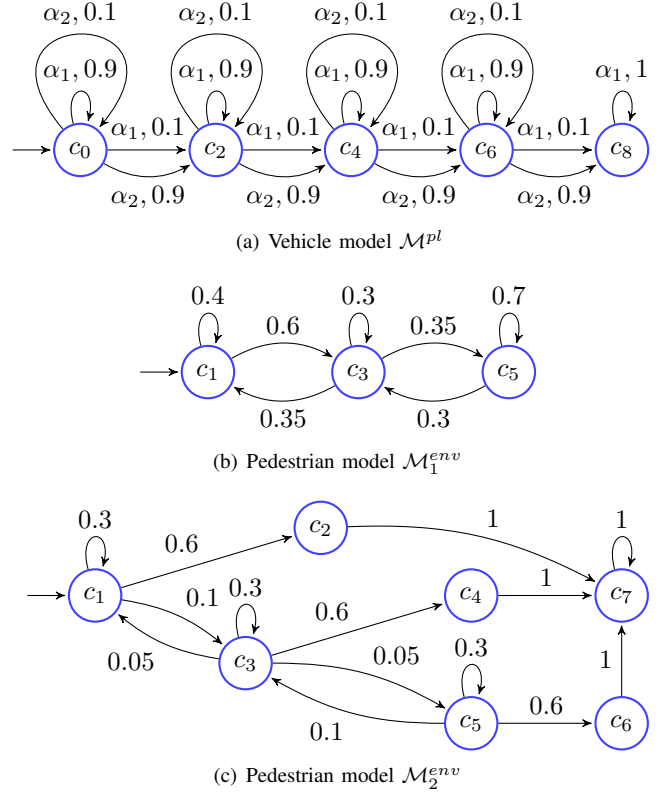

\begin{figure}
\centering 
\begin{tikzpicture}[->,>=stealth',shorten >=1pt,auto,node distance=2cm,bend angle=15]
  \tikzstyle{every state}=[circle,thick,draw=blue!75,minimum size=6mm]

   \node [state] (q0) at (1,0) {$q_0$};
   \node [state] (q1) at (5,0.5){$q_1$};
   \node [state] (q2) at (6,-0.5) {$q_2$};
   
    \draw [ shorten >= 1pt, -> ] (0.1,0) to (q0);
   
   \path (q0) edge [loop above] node {$\neg col \aand \neg c^{pl}_8$} (q0)
                     edge [bend left] node {$c^{pl}_8$} (q1)
                     edge [bend right] node [below] {$\hspace{-3mm}col \aand \neg c^{pl}_8$} (q2)
             (q1) edge [loop above] node {$\true$} (q1)
             (q2) edge [loop above] node {$\true$} (q2);
\end{tikzpicture}
   \caption{A DRA $\A_\varphi$ that recognizes the prefixes of
   $\varphi = \neg col \until c^{pl}_8$ where the collision event $col$ is defined as $col = \bigvee_{j \geq 0} (c^{pl}_j \aand c^{env}_j)$.
   The acceptance condition is $Acc = \{(\emptyset, \{q_1\})\}$.}
   \label{fig:Aphi}
\end{figure}
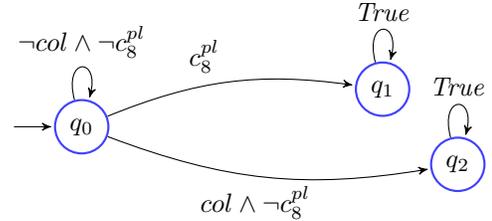

Both the expectation-based and the worst-case-base control policy synthesis
as described in Section \ref{sec:syn-MDP} and \ref{sec:syn-AMDP} is implemented in MATLAB.
The computation was performed on a MacBook Pro with a 2.8 GHz Intel Core 2 Duo processor.

First, we consider the expectation-based control policy synthesis (i.e., Problem \ref{prob:expected}).
As outlined in Section \ref{sec:syn-MDP}, we first construct the MDP that represents the complete system.
After removing all the unreachable states, the resulting MDP contains 65 states and the product MDP
contains 53 states.
The computation time is summarized in Table \ref{tab:comp-time}.
Note that the computation, especially the product MDP construction,
can be sped up significantly if a more efficient representation of DRA is used.
The maximum expected probability for the system to satisfy $\varphi$ is 0.9454.
Examination of the resulting control policy shows that this maximum expected probability of satisfying $\varphi$
can be achieved by applying action $\alpha_1$, i.e., decelerating, until the pedestrian crosses the street 
or the vehicle is not behind the pedestrian in the longitudinal direction, i.e.,
when the vehicle is in cell $c_i$ and the pedestrian is in cell $c_j$ where $j < i$.
(Based on the expectation, the probability that the pedestrian eventually crosses the road is 1 according to the probability measure
defined in (\ref{eq:prob-measure}).)
If we include the belief $\mathbf{B}$ where $\mathbf{B}(\M^{env}_1) = 1$ and $\mathbf{B}(\M^{env}_2) = 0$,
then the expectation-based optimal control policy is such that the vehicle applies $\alpha_1$
until either the pedestrian crosses the road, the vehicle is not behind the pedestrian in the longitudinal direction
or the belief is updated to $\mathbf{B}$, at which point, it applies $\alpha_2$.
Once the vehicle reaches the destination $c_8$, it applies $\alpha_1$ forever.

Next, we consider the worst-case-based control policy synthesis (i.e., Problem \ref{prob:worst-case}).
In this case, the resulting AMDP contains 49 states and the product AMDP contains 53 states after removing
all the unreachable states.
The maximum worst-case probability for the system to satisfy $\varphi$ is 0.9033.
The resulting control policy is slightly more aggressive than the expectation-based policy.
In addition to the cases where the expectation-based controller applies $\alpha_2$,
the worst-case-based controller also applies $\alpha_2$ when the vehicle is in $c_0$ and the pedestrian is in $c_3$
and when the vehicle is in $c_2$ and the pedestrian is in $c_5$.

\begin{table}[h]
\centering
\begin{tabular}{ | l | c | c | c | c | c |}
  \hline                
  & 
  \hspace{-2mm}$\begin{array}{c}\hbox{MDP /}\\ \hbox{AMDP}\end{array}$\hspace{-2mm} & 
  \hspace{-3mm}$\begin{array}{c}\hbox{product}\\ \hbox{MDP / AMDP}\end{array}$\hspace{-3mm} & 
  \hspace{-2mm}$\begin{array}{c}\hbox{Prob}\\ \hbox{vector}\end{array}$\hspace{-2mm} & 
  \hspace{-2mm}$\begin{array}{c}\hbox{Control}\\ \hbox{policy}\end{array}$\hspace{-2mm} & 
  \hspace{-1mm}Total\hspace{-1mm} \\
  \hline
  \hline  
  \hspace{-1mm}Expectation \hspace{-1mm} & 0.05 & 2.31 & 0.73 & 0.08 & 3.17\\
  \hline
  \hspace{-1mm}Worst-case \hspace{-1mm} & 0.20 & 2.73 & 0.46 & 0.05 & 3.44\\
  \hline  
\end{tabular}
\caption{Time required (in seconds) for each step of computation.}
\label{tab:comp-time}
\end{table}

Simulation results are shown in Figure \ref{fig:results-expected} and Figure \ref{fig:results-worst}.
The smaller (red) rectangle represents the pedestrian whereas the bigger (blue) rectangle represents the vehicle.
The filled and unfilled rectangles represent their current positions and the trace of their trajectories, respectively.
Notice that the vehicle successfully reaches its goal without colliding with the pedestrian,
as required by its specification, with the worst-case-based controller being slightly more aggressive.

Finally, we would like to note that due to the structure of this example,
the worst-case-based synthesis problem can be solved without having to deal with the belief space at all.
As the environment cannot be in state $c_2$, $c_4$, $c_6$ or $c_7$ when it is in mode $\M^{env}_1$
and these states only have transitions among themselves,
once the environment transitions to one of these states, we know for sure that it can only be in mode $\M^{env}_2$ and
cannot change its mode anymore.
In states $c_1$, $c_3$ and $c_5$, the environment can be in either mode.
Based on this structure and Remark \ref{remark:AMDP-construction}, 
we can construct an AMDP that represents the complete system with smaller number of states
than $\M^\A$ constructed using the method described in Section \ref{ssec:AMDP-construction}.
Exploiting the structure of the problem to reduce the size of AMDP is subject to future work.

\begin{figure}
\subfigure[$t=0$]{
	\centering\includegraphics[trim=3.5cm 10cm 2.5cm 10cm, clip=true, width=0.22\textwidth, height=1.8cm]{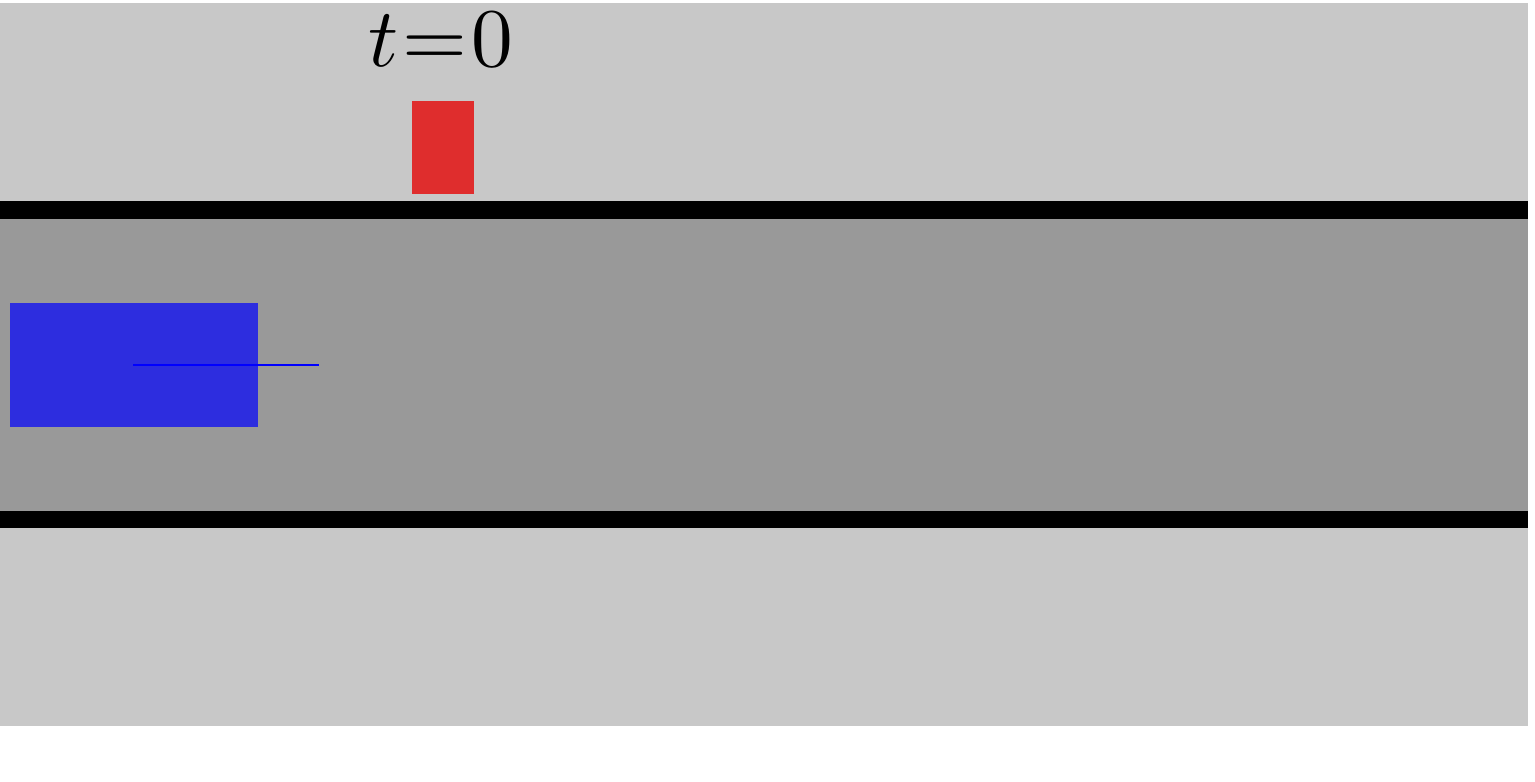}
}
\hfill
\subfigure[$t=1$]{
	\centering\includegraphics[trim=3.5cm 10cm 2.5cm 10cm, clip=true, width=0.22\textwidth, height=1.8cm]{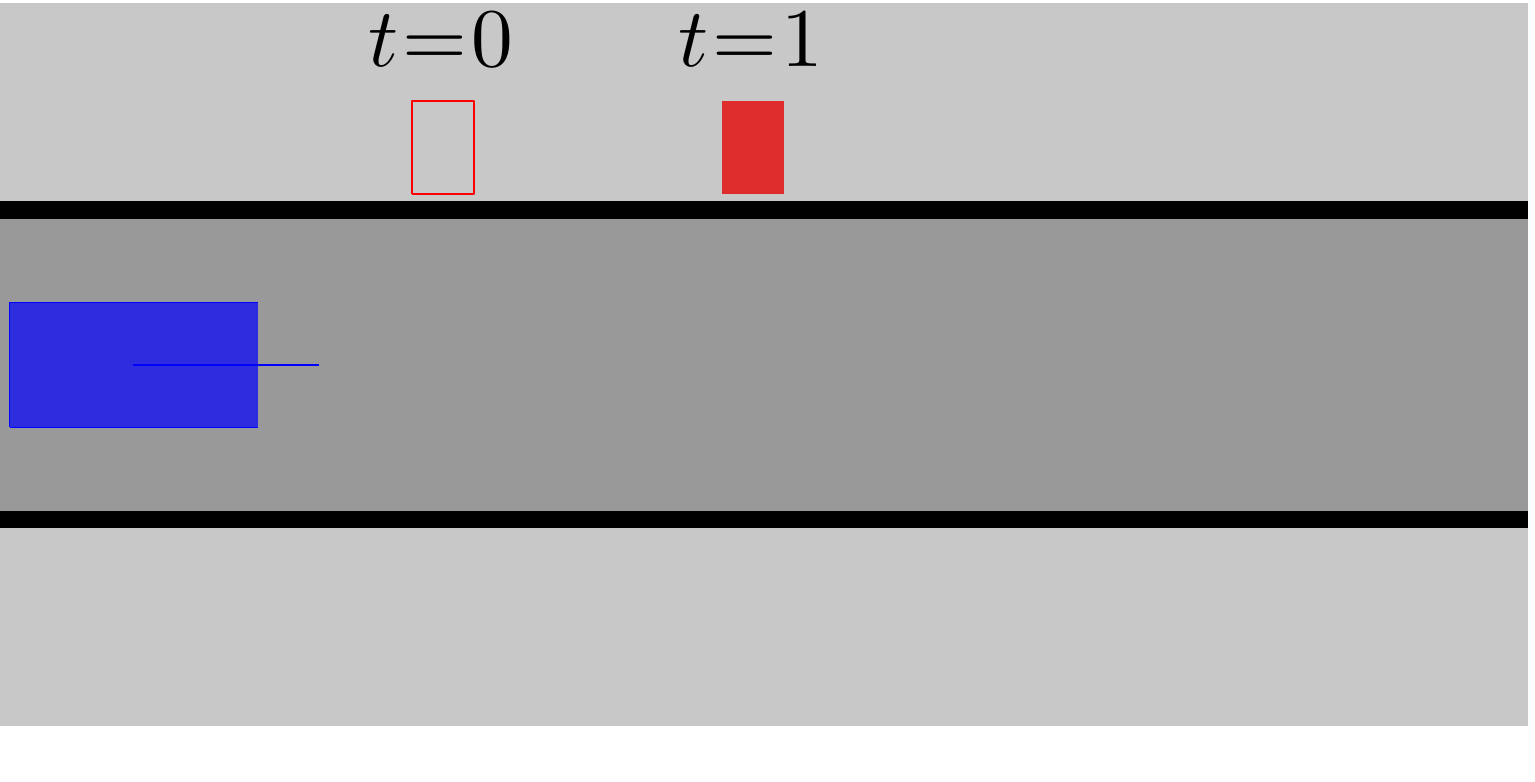}
}
\hfill
\subfigure[$t=2$]{
	\centering\includegraphics[trim=3.5cm 10cm 2.5cm 10cm, clip=true, width=0.22\textwidth, height=1.8cm]{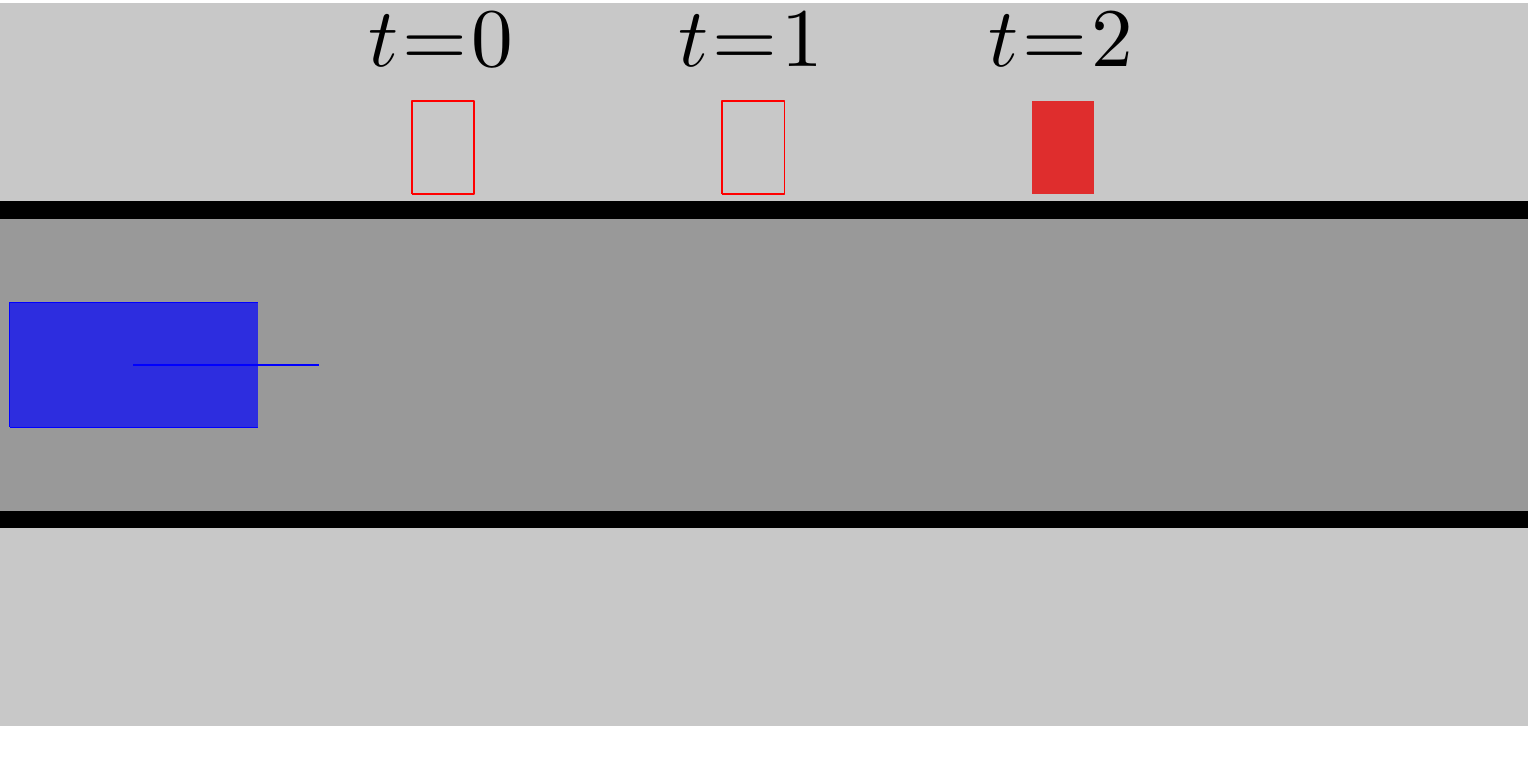}
}
\hfill
\subfigure[$t=3$]{
	\centering\includegraphics[trim=3.5cm 10cm 2.5cm 10cm, clip=true, width=0.22\textwidth, height=1.8cm]{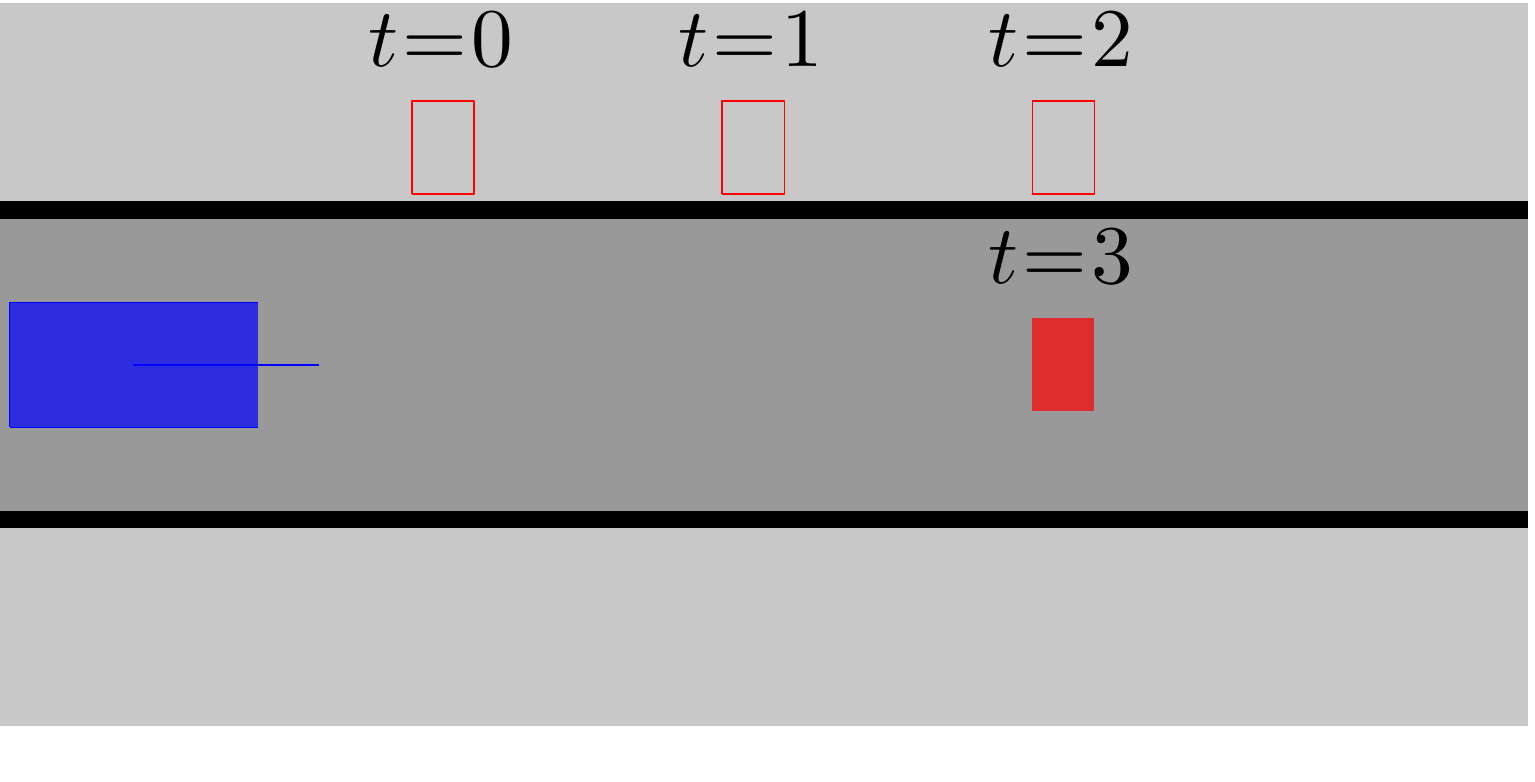}
}
\hfill
\subfigure[$t=4$]{
	\centering\includegraphics[trim=3.5cm 10cm 2.5cm 10cm, clip=true, width=0.22\textwidth, height=1.8cm]{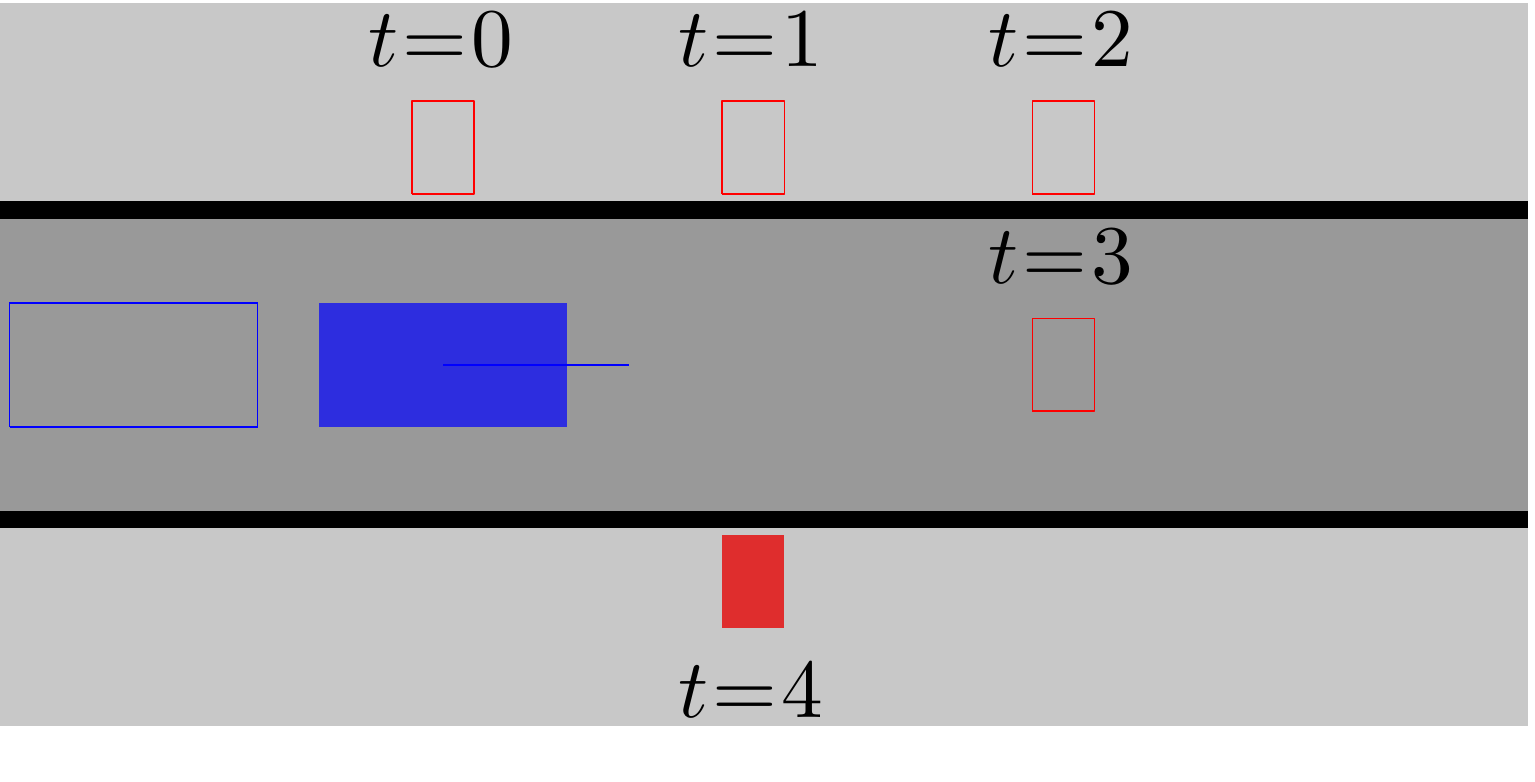}
}
\hfill
\subfigure[$t=5$]{
	\centering\includegraphics[trim=3.5cm 10cm 2.5cm 10cm, clip=true, width=0.22\textwidth, height=1.8cm]{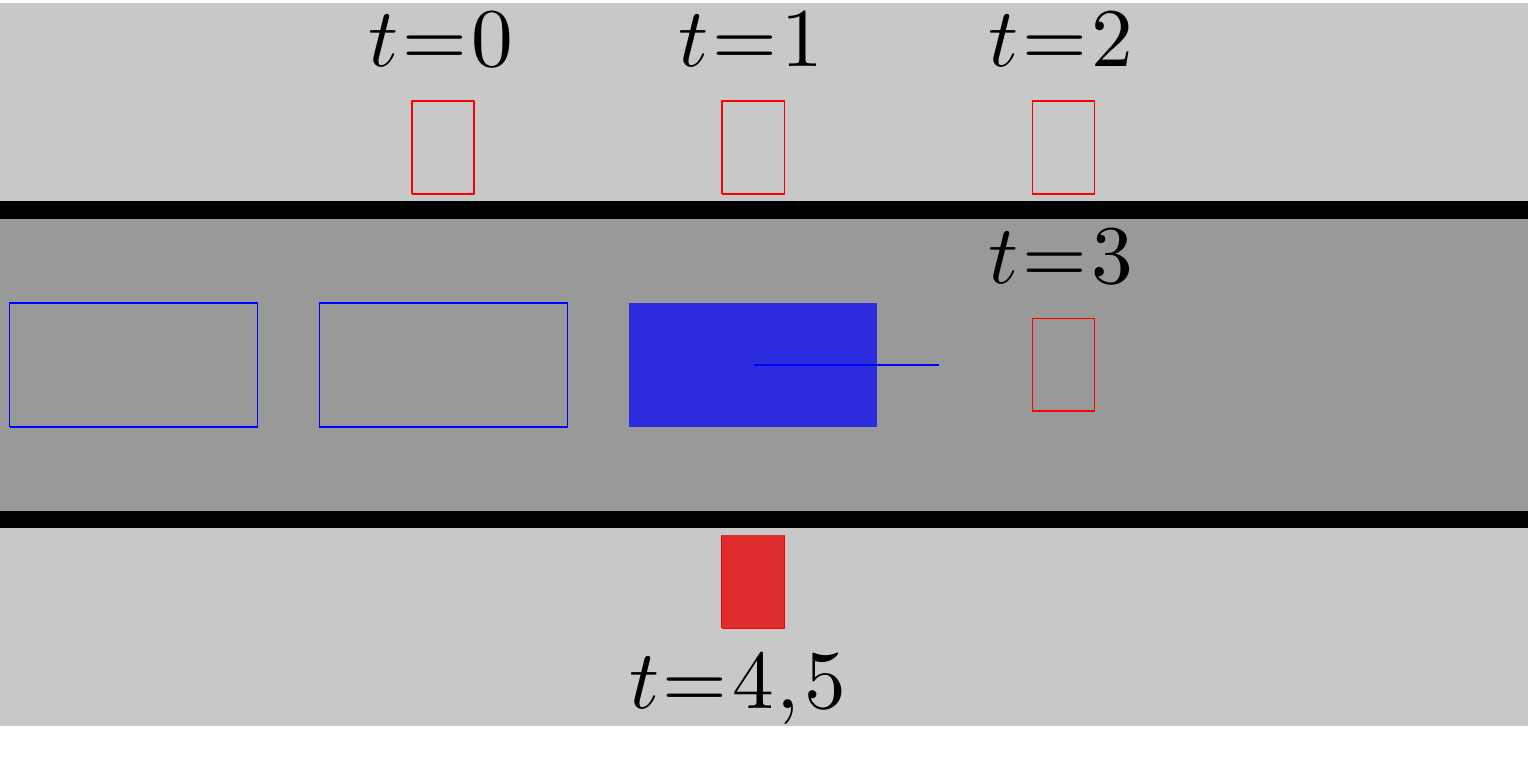}
}
\hfill
\subfigure[$t=6$]{
	\centering\includegraphics[trim=3.5cm 10cm 2.5cm 10cm, clip=true, width=0.22\textwidth, height=1.8cm]{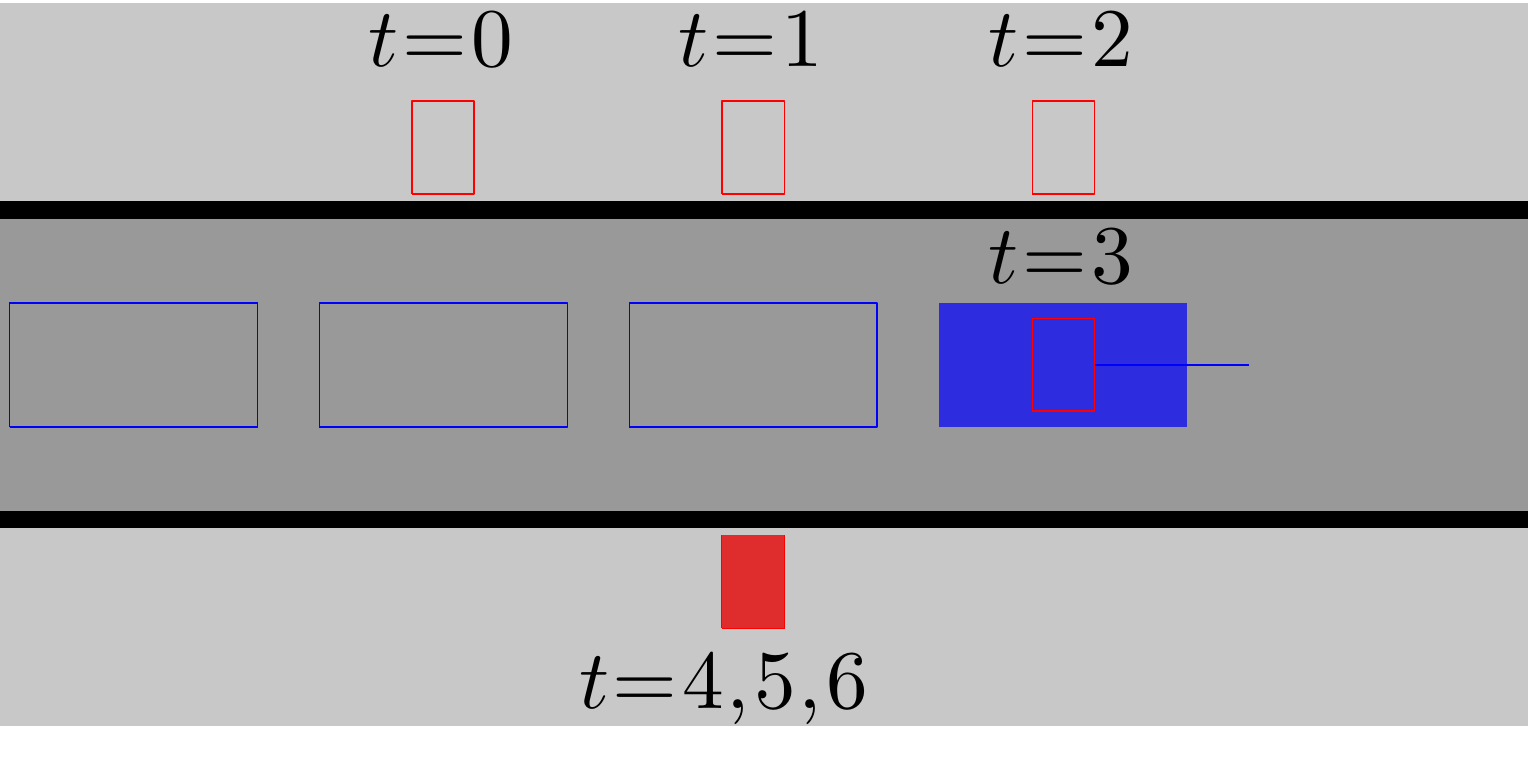}
}
\hfill
\subfigure[$t=7$]{
	\centering\includegraphics[trim=3.5cm 10cm 2.5cm 10cm, clip=true, width=0.22\textwidth, height=1.8cm]{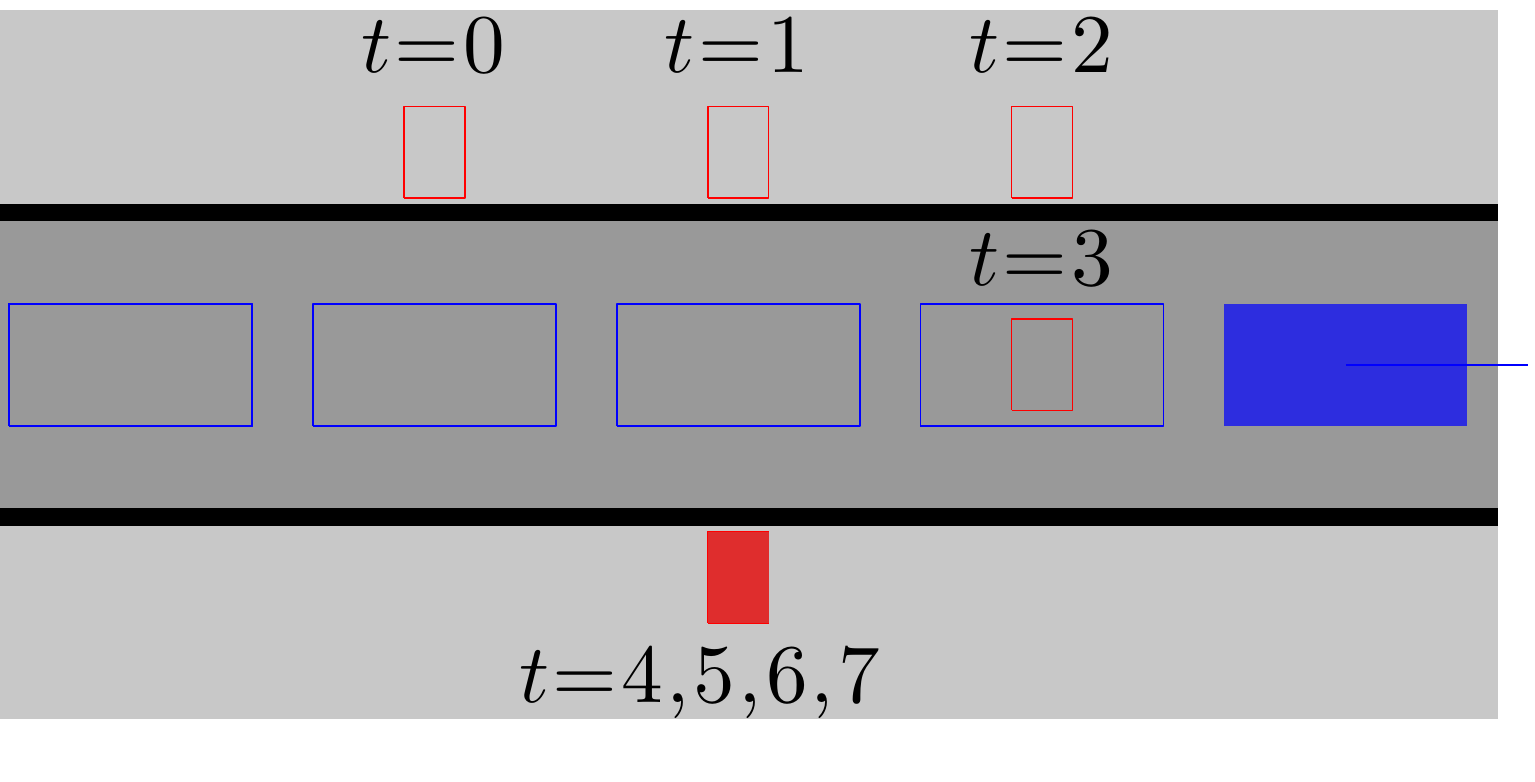}
}
\caption{Expectation-based control policy. 
At time $0 \leq t < 3$, the vehicle applies $\alpha_1$.
$\alpha_2$ is applied at time $3 \leq t < 7$, after which the vehicle reaches the goal and applies $\alpha_1$ forever. }
\label{fig:results-expected}
\end{figure}

\begin{figure}
\subfigure[$t=0$]{
	\centering\includegraphics[trim=3.5cm 10cm 2.5cm 10cm, clip=true, width=0.22\textwidth, height=1.8cm]{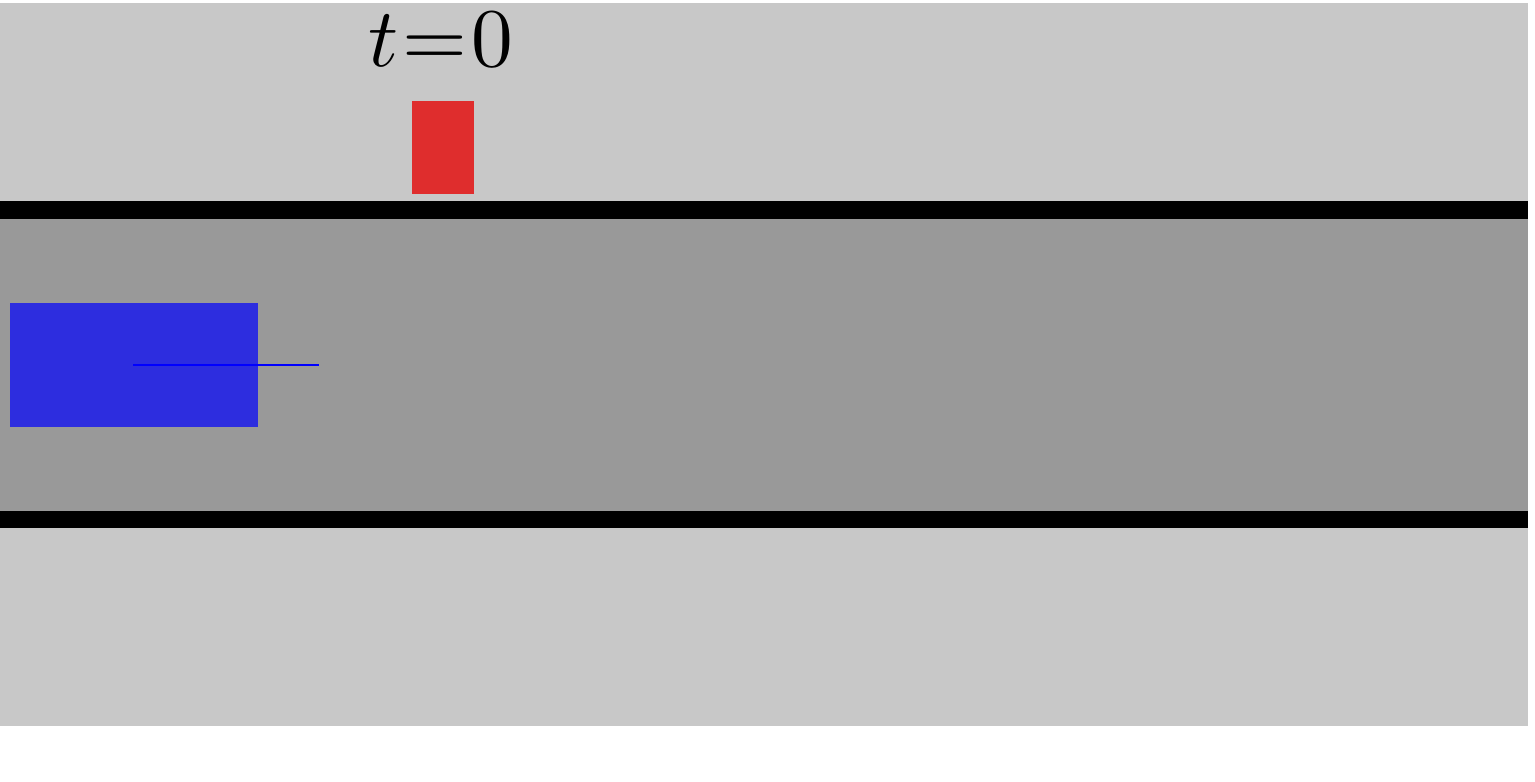}
}
\subfigure[$t=1$]{
	\centering\includegraphics[trim=3.5cm 10cm 2.5cm 10cm, clip=true, width=0.22\textwidth, height=1.8cm]{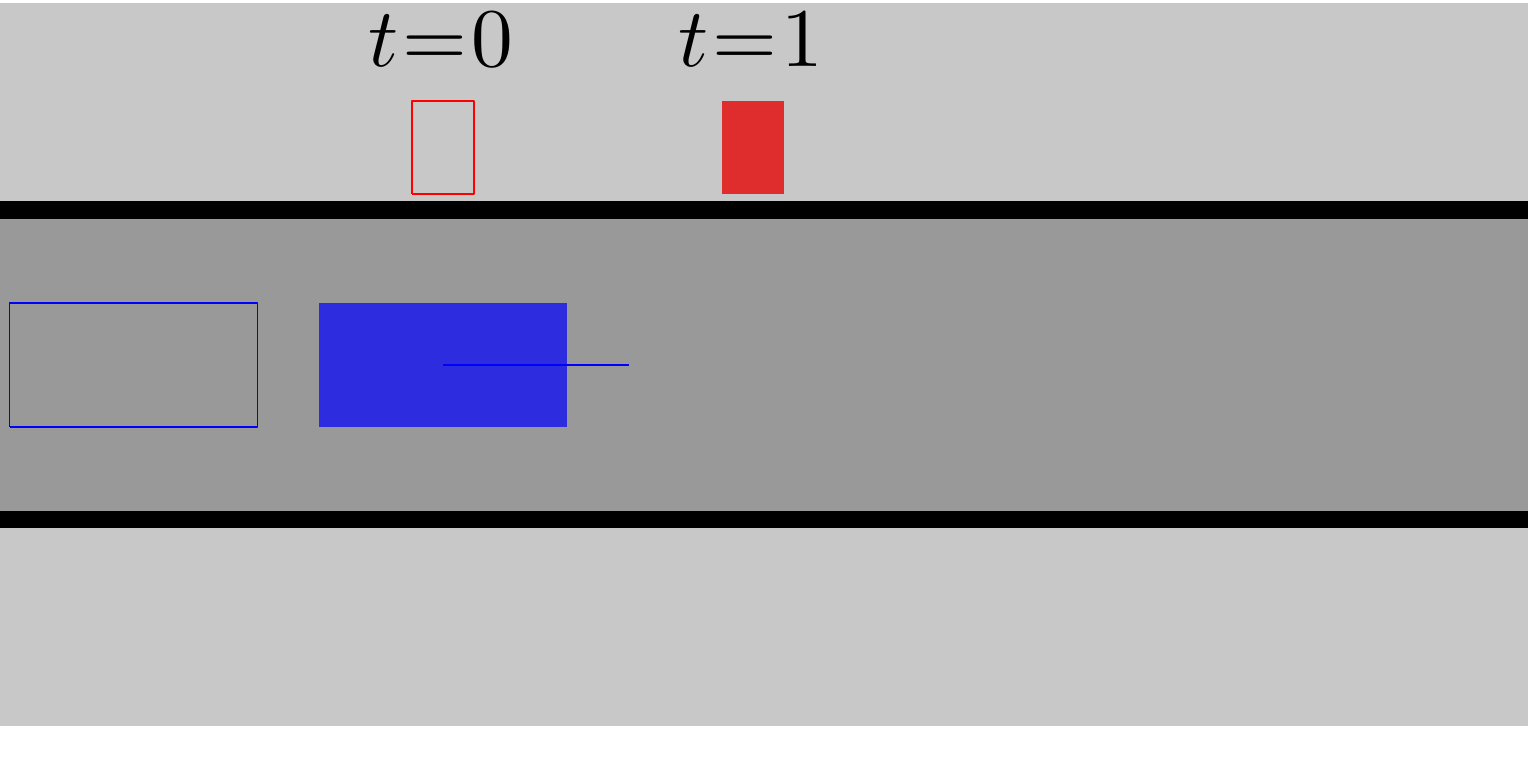}
}
\subfigure[$t=2$]{
	\centering\includegraphics[trim=3.5cm 10cm 2.5cm 10cm, clip=true, width=0.22\textwidth, height=1.8cm]{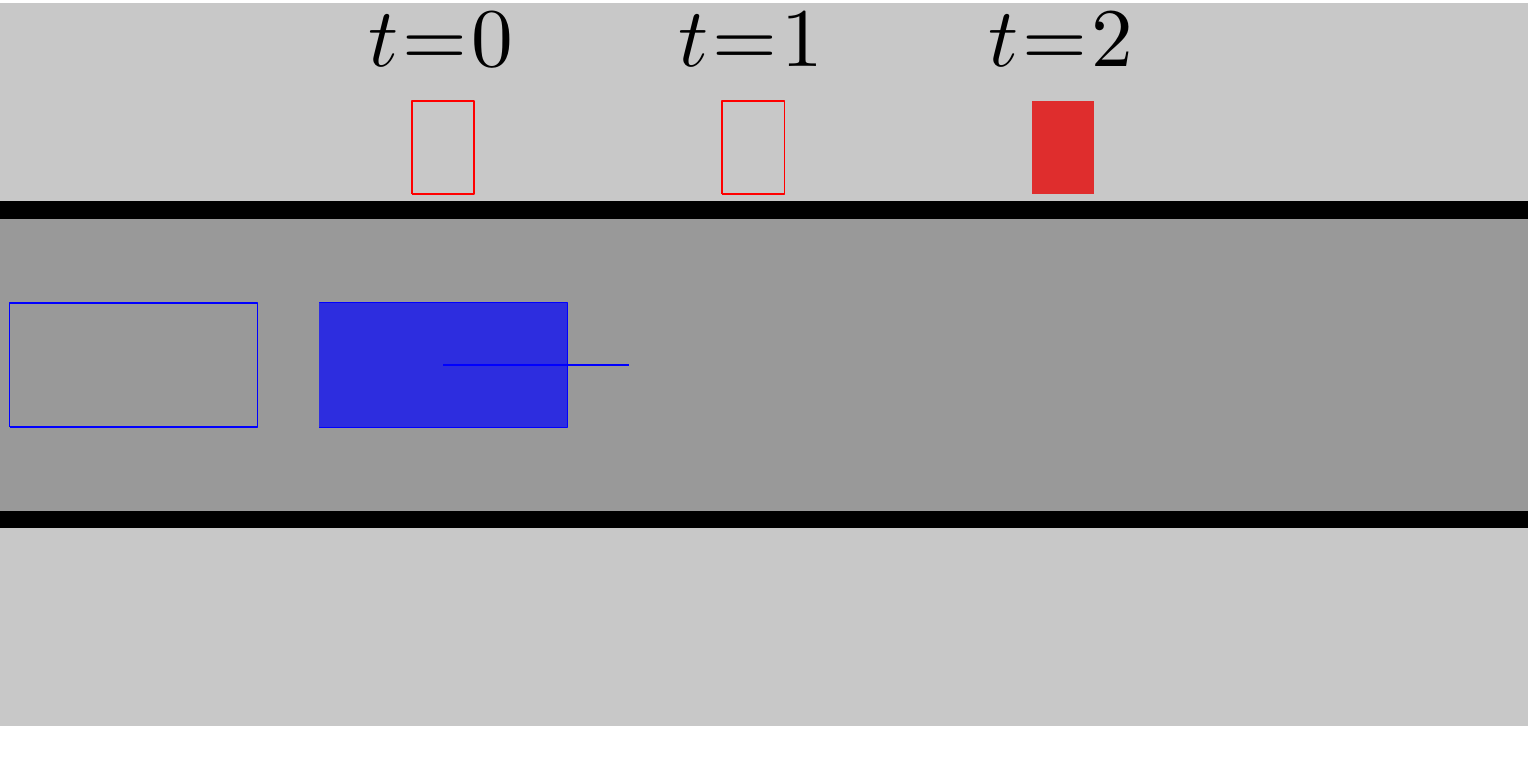}
}
\hfill
\subfigure[$t=3$]{
	\centering\includegraphics[trim=3.5cm 10cm 2.5cm 10cm, clip=true, width=0.22\textwidth, height=1.8cm]{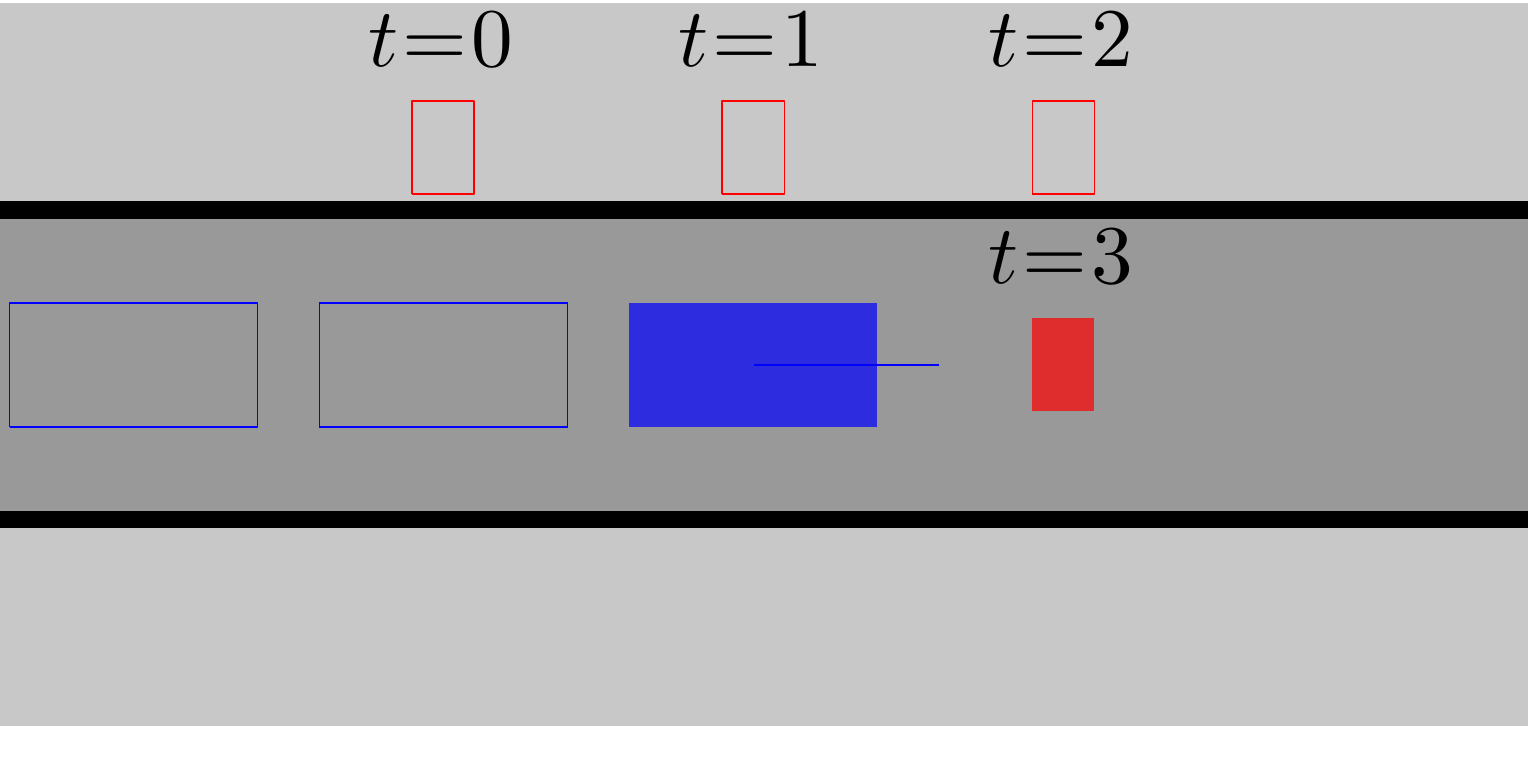}
}
\hfill
\subfigure[$t=4$]{
	\centering\includegraphics[trim=3.5cm 10cm 2.5cm 10cm, clip=true, width=0.22\textwidth, height=1.8cm]{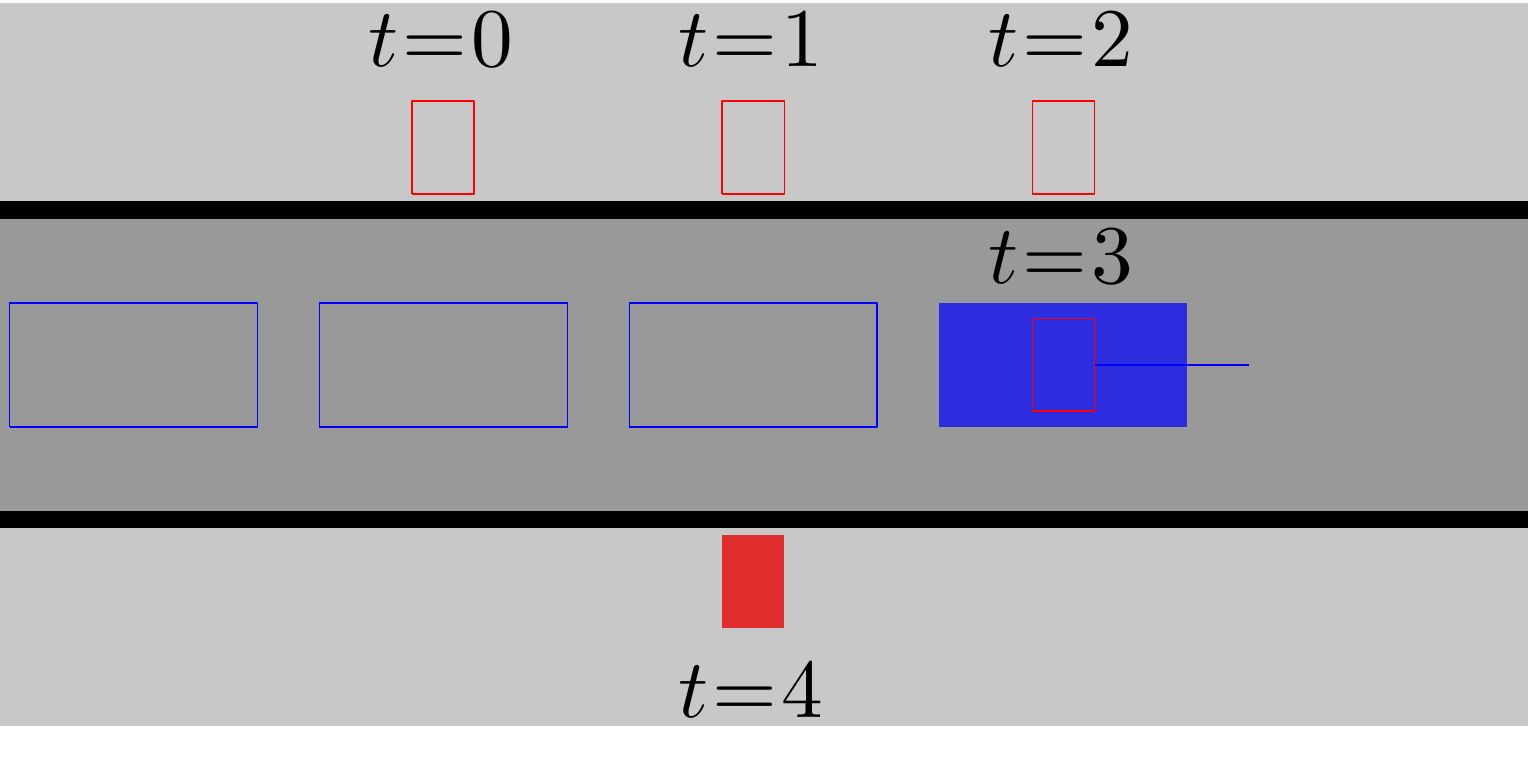}
}
\hfill
\subfigure[$t=5$]{
	\centering\includegraphics[trim=3.5cm 10cm 2.5cm 10cm, clip=true, width=0.22\textwidth, height=1.8cm]{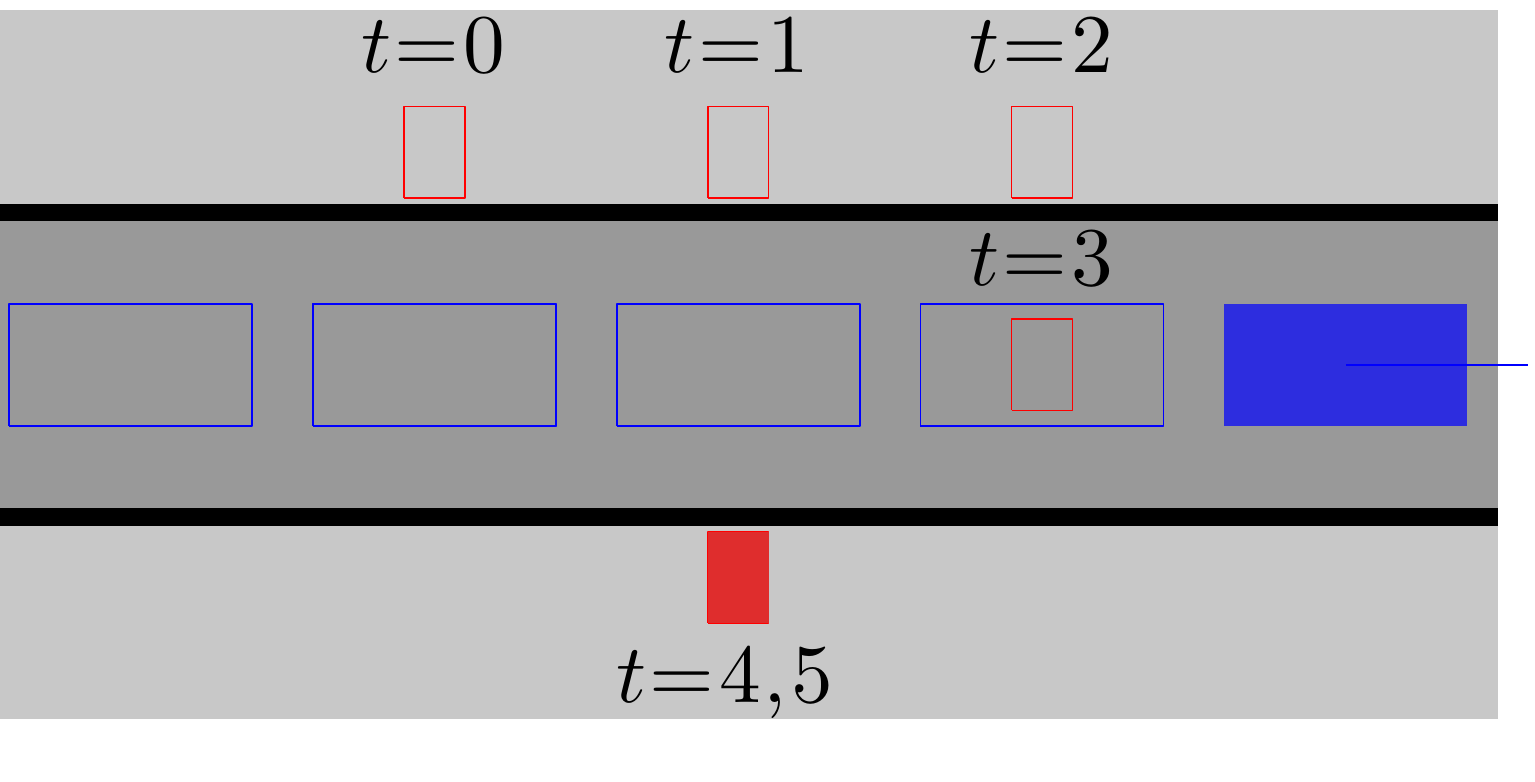}
}
\hfill
\subfigure[$t=6$]{
	\centering\includegraphics[trim=3.5cm 10cm 2.5cm 10cm, clip=true, width=0.22\textwidth, height=1.8cm]{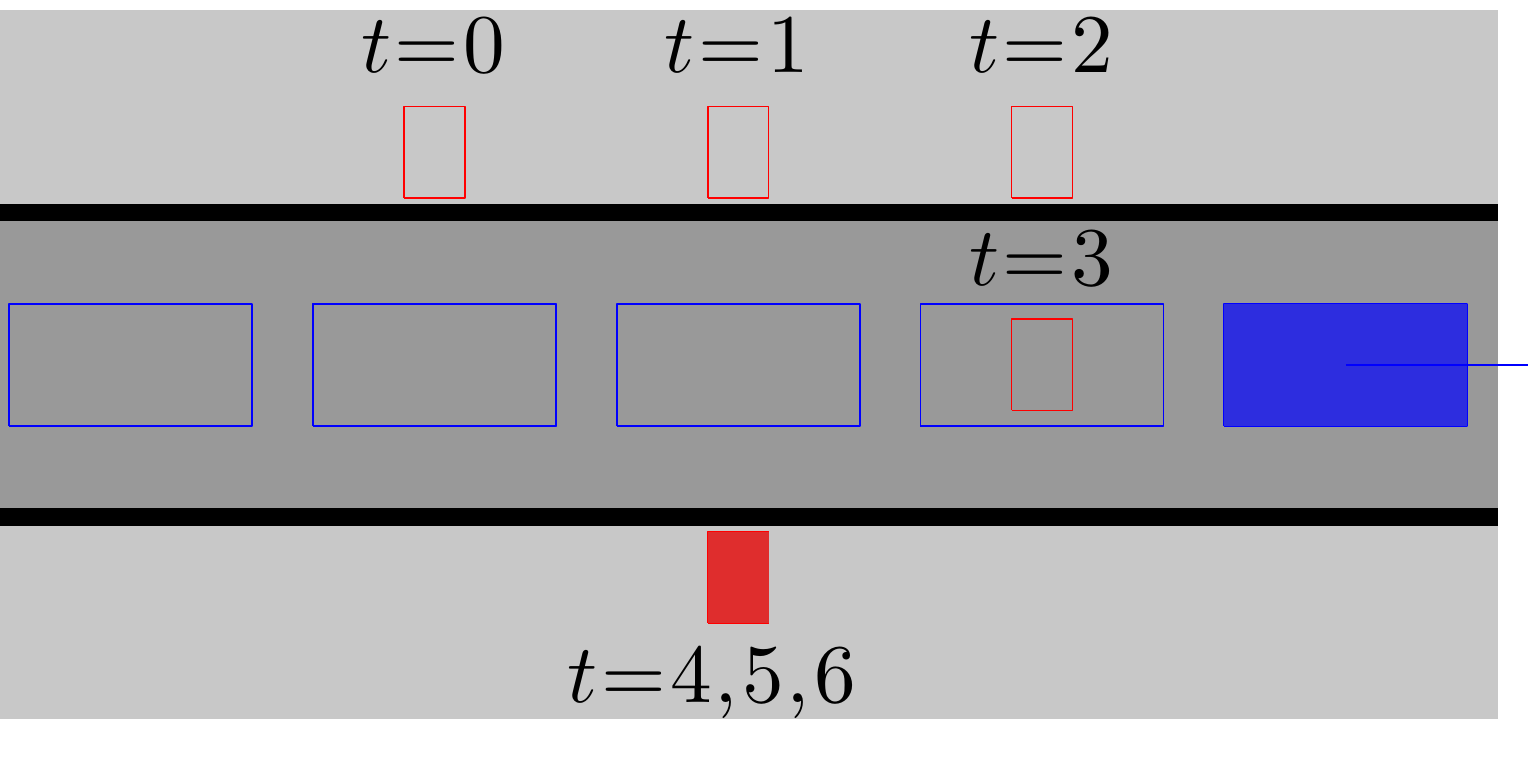}
}
\hfill
\subfigure[$t=7$]{
	\centering\includegraphics[trim=3.5cm 10cm 2.5cm 10cm, clip=true, width=0.22\textwidth, height=1.8cm]{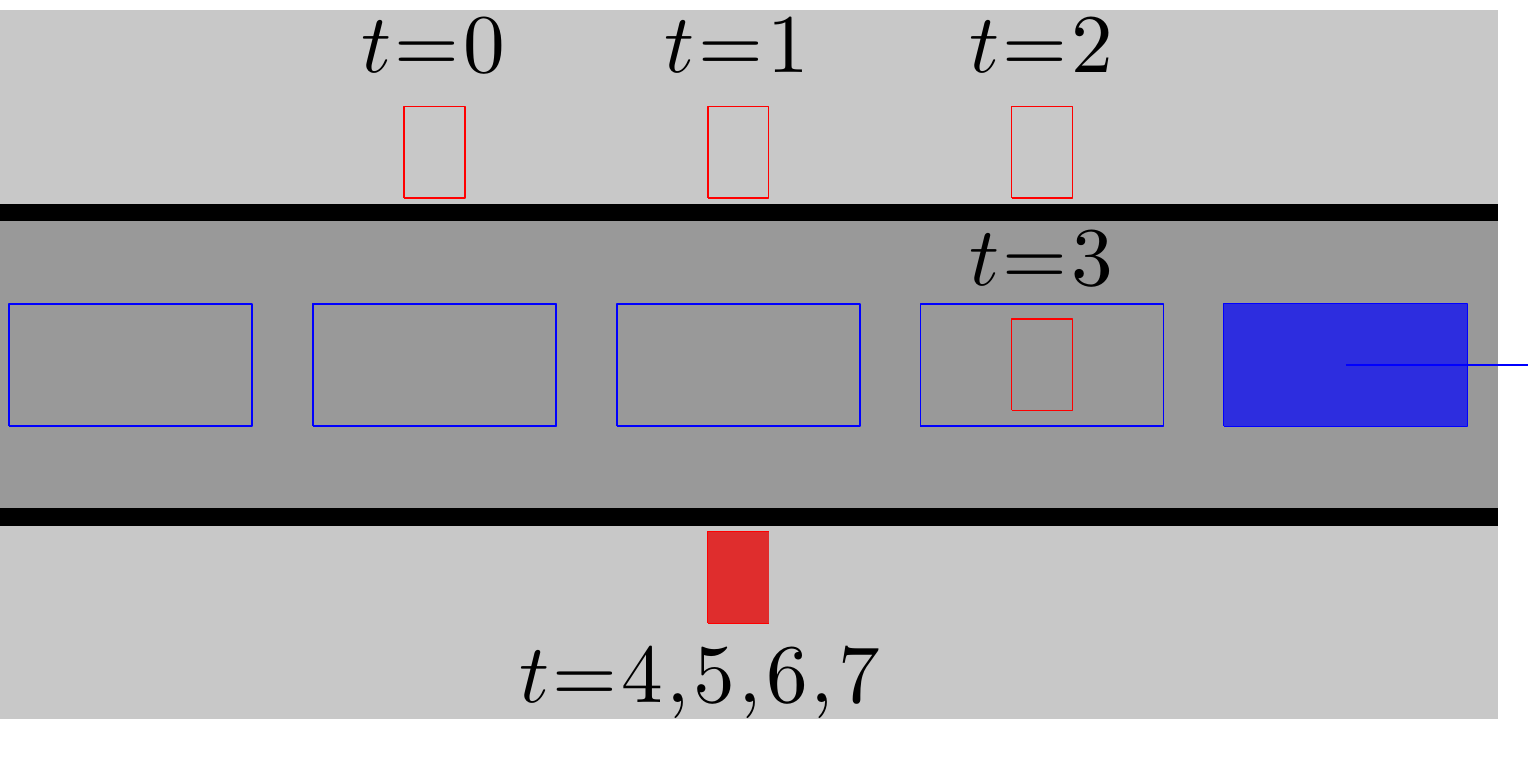}
}
\caption{Worst-case-based control policy. 
At time $0 \leq t < 2$, the vehicle applies $\alpha_1$. 
However, the vehicle moves forward during time $0 \leq t < 1$ because of the uncertainties in the vehicle model.
$\alpha_2$ is applied at time $2 \leq t < 5$, after which the vehicle reaches the goal and applies $\alpha_1$ forever. }
\label{fig:results-worst}
\end{figure}

\section{Conclusions and Future Work}
\label{sec:conclusions}
We took an initial step towards solving POMDPs that are subject to temporal logic specifications.
In particular, we considered the problem where the system interacts with its dynamic environment.
A collection of possible environment models are available to the system.
Different models correspond to different modes of the environment.
However, the system does not know in which mode the environment is.
In addition, the environment may change its mode during an execution.
Control policy synthesis was considered with respect to two different objectives:
maximizing the expected probability and maximizing the worst-case probability
that the system satisfies a given temporal logic specification.

Future work includes investigating methodologies to approximate the belief space
with a finite set of representative points.
This problem has been considered extensively in the POMDP literature.
Since the value iteration used to obtain a solution to our expectation-based synthesis problem 
is similar to the value iteration used to solve POMDP problems where the expected reward is to be maximized,
we believe that existing sampling techniques used to solve POMDP problems can be adapted to solve our problem.
Another direction of research is to integrate methodologies for discrete state estimation
to reduce the size of AMDP for the worst-case-based synthesis problem.

\section*{Acknowledgments}
The authors gratefully acknowledge Tirthankar Bandyopadhyay for inspiring
discussions.

\bibliographystyle{ieeetr}
\bibliography{ref}

\end{document}